\documentclass{article}
\usepackage{matthew} 
\graphicspath{./img/}
\usepackage{enumitem}
\usepackage{mathrsfs}

\graphicspath{{img/}}

\title{Cycle Patterns and Mean Payoff Games}
\author{Georg Loho\thanks{FU Berlin, Germany \& U Twente, The Netherlands.} 
\and Matthew Maat\thanks{U Twente, The Netherlands}
\and Mateusz Skomra\thanks{LAAS-CNRS, Université de Toulouse, CNRS, Toulouse, France.}}
\date{}

\newcommand{\Cs}{\mathscr{C}}

\newcommand{\RR}{\mathbb{R}}
\newcommand{\ZZ}{\mathbb{Z}}
\newcommand{\QQ}{\mathbb{Q}}

\newcommand{\sgn}{\operatorname{sgn}}
\newcommand{\val}{\mathrm{val}}
\newcommand{\region}[1]{\mathcal{Z}(#1)}
\newcommand{\dectree}{\mathcal{T}}
\newcommand{\Nodesub}{\mathscr{N}}
\newcommand{\Win}{V^{+}}

\newcommand{\Vmax}{V_{\text{Max}}}
\newcommand{\Vmin}{V_{\text{Min}}}

\begin{document}
\maketitle
\begin{abstract}
We introduce the concept of a \emph{cycle pattern} for directed graphs 
as functions from the set of cycles to the set $\{-,0,+\}$.   
The key example for such a pattern is derived from a weight function, giving rise to the sign of the total weight of the edges for each cycle. 
Hence, cycle patterns describe a fundamental structure of a weighted digraph, and they arise naturally in games on graphs, in particular parity games, mean payoff games, and energy games. 
        
Our contribution is threefold: we analyze the structure and derive hardness results for the realization of cycle patterns by weight functions. 
Then we use them to show hardness of solving games given the limited information of a cycle pattern. 
Finally, we identify a novel geometric hardness measure for solving mean payoff games (MPG) using the framework of linear decision trees, and use cycle patterns to derive lower bounds with respect to this measure, for large classes of algorithms for MPGs.  
\end{abstract}
\newcommand{\setOf}[2]{\{#1 \colon #2\}}

\section{Introduction}

The set of negative cycles is a basic structure of a directed graph equipped with a weight function. 
This set is also the main obstacle for shortest path algorithms, and it serves to determine the winner in certain games on graphs.
Furthermore, it is a special case of an infeasible subsystem of a linear inequality system, and it has been used to deduce the hardness of determining all vertices of a polyhedron.  
We consider a refined version of the structure captured by the set of negative cycles, namely a function from the directed cycles in a digraph to $\{-,0,+\}$, depending on the signs of the weights of the cycles. 
More generally, this leads to the concept of a \emph{cycle pattern}, an arbitrary function on the set of directed cycles of a directed graph to the set $\{-,0,+\}$. 
To the best of our knowledge, although this concept is natural and easy to define, it has not been studied before.

Our goal in studying cycle patterns is twofold. 
On the one hand, we develop their basic structural theory and address fundamental hardness questions. 
On the other hand, we use cycle patterns as a framework to analyze the complexity of games on graphs. 

The first motivation to study cycle patterns stems from the analogy with 
another combinatorial concept: oriented matroids capture the essence of patterns arising in an algorithmic context. 
Specifically, oriented matroids generalize the structure of the signs of the minors of a matrix or the conformal orientations of spanning trees of a graph. Although a cycle pattern is more basic in nature, from an algorithmic point of view, several similar questions arise. 
For example, a cycle pattern is \emph{realizable} if it is associated with a weighted directed graph and captures which cycles have a positive, zero or a negative weight. 
Here, it is remarkable that a purely combinatorial criterion characterizes the realizable cycle patterns (\cref{thm:char1}). 
This is in contrast to oriented matroids for which realizability cannot be checked by a purely combinatorial criterion. 
However, the algorithmic problem of deciding the realizability of a cycle pattern given by an arithmetic circuit is hard (\cref{thm:realizablecomplexity}) which shows that cycle patterns still capture a rich structure. 
Furthermore, the range of integers required to realize all cycle patterns is big, 
as demonstrated by rather tight upper and lower bounds (\cref{thm:char1} and \cref{thm:expweights}). 
Finally, our results on the realizability problem reveal more general insights into the size of integer points in polyhedral fans in geometric combinatorics. 

\subsection{Structure theory for games on graphs}

The second main motivation for this paper is the problem of solving mean payoff games (MPGs), energy games (EGs) and parity games (PGs). 
These games are zero-sum two player games played on a directed graph, and they are closely related to each other: many problems related to mean payoff games and energy games are equivalent, and solving parity games can be reduced to solving mean payoff games. The decision problems related to solving these games are known to be in NP$\cap$coNP. 
However, the question whether there exists a polynomial-time algorithm for them has been open for decades.
We refer to \cite{fijalkow2023gamesgraphs} for a broad overview of related classes of games and their algorithmic complexity, including further references. 

Some advances on this question were made in the form of quasipolynomial parity game algorithms and a number of pseudopolynomial algorithms for mean payoff games and energy games. 
However, for the majority of these algorithms, examples have been constructed with running time close to the upper bound (quasipolynomial or exponential running time).

Many of these worst-case instances are handcrafted and considered as rare, they are usually very algorithm-specific, where one algorithm can often solve the hard instances for another. 
This is reinforced by the existence of a smoothed polynomial algorithm for mean payoff games \cite{loff_et_al_SmoothedGames}; 
the smoothed polynomial algorithm is fast on every game that is sufficiently far away from a certain set of ``non-generic'' game instances. 
We use cycle patterns to give a new quantification of which structures are hard independent of the specific algorithms. 
Cycle patterns capture the most fundamental combinatorial structure appearing in mean payoff games, as the winning condition of MPGs
is formulated in terms of the signs of cycles. 
On this level, one can consider a game as \emph{hard} if the derived cycle pattern has only exponential realizations, that is, all weight functions leading to the same cycle pattern use entries that exponential in the size of the game graph.
In contrast, a game with weights polynomial in the size of the graph can be solved in polynomial time by any of the existing pseudopolynomial algorithms. 
To apply this measure of complexity to large classes of algorithms we use the framework of \emph{linear decision trees} (\cref{sec:geometric}). 

As cycle patterns already capture a crucial part of the structure of an MPG, we investigate how much of the information of a weight function is needed to solve an MPG efficiently. 
We study this in two directions. 
On one hand, we consider potential preprocessing of an MPG to reduce the size of the weights and make algorithms more efficient. 
This is inspired by the seminal result~\cite{FrankTardos:1987} that uses a preprocessing preserving the structure of a problem and thereby making certain algorithms in combinatorial optimization strongly polynomial. 
It turns out that this does not work well in the context of MPGs while preserving the cycle pattern. 
On the other hand, we deduce how hard it becomes to solve an MPG when only the information of the cycle pattern is available. 
For this, we choose Boolean circuits as a suitable formalism to encode a cycle pattern. 
Specifically, assuming $RP\neq NP$, there cannot be a randomized polynomial time algorithm solving mean payoff games using only the cycle pattern (\cref{sec:solving-games-via-cycle-patterns}). 
We also use the encoding by Boolean circuits to understand complexity questions on realizations of cycle patterns (\cref{sec:complexity}). 
We show that checking realizability is coNP-complete, and that checking parity-realizability (i.e. whether a cycle pattern can be realized in a parity game) is coNP-complete even if a realization is given. 
Furthermore, distinguishing cycle patterns with given realizations is NP-complete. 

Finally, going back to the analogy with oriented matroids, we also draw motivation from the theory of polyhedra as a structure theory of linear programming. 
The decision problem for MPGs is equivalent to \emph{tropical linear programming}~\cite{AkianGaubertGuterman:2012} which leads to the corresponding structure theory of tropical polyhedra; see~\cite{Joswig:2022} for more in this direction.  
While the study of \emph{tropical oriented matroids}~\cite{ArdilaDevelin:2009} and \emph{abstract tropical linear programming}~\cite{Loho:2016} form a further motivation for introducing cycle patterns as another tool to analyze tropical linear inequality systems, the formal connection is not directly clear.

\subsection{Related work}

The problem of computing the set of all negative weight cycles has been the object of limited study. 
An algorithm to do so was developed in \cite{YamadaKinoshita:2002}, although it has been shown that it is NP-hard to enumerate all negative cycles \cite{KhachiyanBorosBorysElbassioniGurvich:2008}. 
The latter result was used to show NP-hardness for several problems for polyhedra \cite{BorosElbassioniGurvichTiwary:2011}. 
Another well-studied and related topic is that of cycle bases. 
Associating a (signed) incidence vector to each cycle of an (undirected) graph leads to the \emph{cycle space} formed by their linear span. 
Finding ``good'' bases of this space, like minimum weight w.r.t. a weight function, is a useful tool in several algorithms; we refer to \cite{KAVITHA2009199} for further details. 
Although the cycles in our context are usually directed and with an additional sign information, one still can define an analogue of a cycle space as discussed in \cref{sec:cyclespace}.

Parity games, mean payoff games and energy games are well-studied because of their interesting complexity status, and their occurrence in applications like model checking \cite{EMERSON2001491}, and-or-scheduling \cite{mohring_scheduling_2004} and hyper temporal networks \cite{Comin2014Hypertemporal}. 
Many algorithms have been developed that solve these games; we give a brief overview of the occurring approaches. 

First of all, strategy improvement is an algorithm that can be used to solve classes of games far broader than the games we discuss here. Its main idea is to start with some arbitrary strategy, find some valuation of the strategy, and augment the strategy with improving moves until optimality is reached. There are versions of strategy improvement specialized for mean payoff games \cite{DBLP:journals/dam/BjorklundV07, Dhingra2006HowGames}, energy games \cite{DBLP:journals/ijfcs/BrimC12} and parity games \cite{DBLP:conf/cav/VogeJ00}. Under some conditions, strategy improvement algorithms can be made symmetric for the two players \cite{DBLP:conf/icalp/ScheweTV15}.
Another approach that works for many types of games is value iteration. Its main idea is to maintain some value on each vertex of the graph, and then updating the values until some fixed point is reached. Value iteration also has versions specialized for mean payoff games \cite{ZWICK1996343}, energy games \cite{DBLP:journals/fmsd/BrimCDGR11} (a fast version is in \cite{dorfman_et_al:LIPIcs.ICALP.2019.114,austin2023erratatofasterdeterministic}) and for parity games \cite{Jurdzinski2017Succinct}. The latter can solve parity games in quasipolynomial time, using a structure called a universal tree.

There are also some more specialized algorithms for specific games. For example, the GKK algorithm \cite{gurvich1990cyclic} for mean payoff games augments potentials on vertices. It also has a symmetric version \cite{Ohlmann2022GKK}. Since solving mean payoff games is equivalent to some problems related to tropical linear programs, it is also possible to solve these games with a combinatorial simplex method \cite{Allamigeon2014TropicalSimplex}. For parity games, there are attractor-based algorithms like Zielonka's algorithm \cite{ZIELONKA1998135} or its quasipolynomial variants \cite{parys:LIPIcs.MFCS.2019.10}. Finally, there are parity game algorithms based on $\mu$-calculus or automata, like the first quasipolynomial algorithm \cite{Calude2017Quasipolynomial}.

To find some order in this sea of ideas and algorithms, some attempts have been made to highlight the underlying structure of some classes of algorithms for games on graphs. 
Strategy improvement has been considered more generally on lattices \cite{baldan_et_al:LIPIcs.CSL.2023.7,KohLoho:2021}.
Another related concept that describes the underlying structure of value iteration is universal graphs \cite{fijalkow:hal-03800510}. 
This generalizes the concept of universal trees in parity games. Even more general than these is the concept of monotonic graphs \cite{ohlmann2021monotonic}, which captures the essence of means to measure progress in strategy improvement, value iteration and attractor decompositions.

\subsection{Outline of the paper}
This paper is divided into two main parts. 
In the first part, we focus only on the graph-theoretic aspects. 
In \cref{sec:cycle_patterns}, we introduce the notion of cycle pattern, and discuss realizability and parity-realizability. 
Then we discuss the (lower) bounds for realizations.  
We also discuss the complexity of related problems in \cref{sec:complexity}.

The second part focuses on applying the concepts and results to the study of solving games on graphs. 
In \cref{sec:games} we remind the reader of the rules of these games, 
and show hardness of the latter problem in the cycle pattern framework. 
Finally, in \cref{sec:geometric}, combining a geometric point of view with the framework of linear decision trees, we draw geometric and algorithmic implications from the theory developed before. 
This highlights the difficulty of some underlying problems of solving games on graphs.

\section*{Acknowledgements}

We thank Antonio Casares, Nathana\"{e}l Fijalkow, Michael Joswig and Cedric Koh for insightful conversations, and Bruno Loff especially for explanations about the linear decision tree model.
This research benefited from the support of the FMJH Program PGMO. 

\section{Cycle patterns}
\label{sec:cycle_patterns}

We introduce the combinatorial structure of \emph{cycle patterns} and examine how this relates to the sets of positive and negative cycles in a weighted digraph. 

\subsection{Basics}

We collect basic terminology and the fundamental definitions. 

\subsubsection{Terminology}\label{sec:terminology}

Let $G = (V,E)$ be a directed graph, potentially with parallel edges and loops. 
While for games one can usually get rid of parallel edges and loops via preprocessing, we leave as much flexibility in the definition of a graph as possible. 
We will always assume that $G$ is strongly connected to simplify proofs.
We set $n = |V|$ and $m = |E|$ unless stated otherwise. We refer to the vertices of $G$ or a subgraph $H$ by $V(G)$ and $V(H)$, respectively.
A \emph{walk} is an alternating sequence $v_0, e_1, v_1, e_2, \ldots, v_k$ of vertices and directed edges such that for $1 \leq i \leq k$, the edge $e_i$ goes from $v_{i-1}$ to $v_i$. 
A \emph{(simple) cycle} is a walk such that each node has at most one in-going and at most one out-going edge and $v_0 = v_k$. 
For convenience of notation, we often regard a cycle as just a sequence of edges or a sequence of vertices.
For each subset $S$ of the edges of $G$ we denote its characteristic vector by $\chi(S) \in \{0,1\}^E$. 
Given a cycle $C$, we denote by $\chi(C)$ the characteristic vector of the set of edges of $C$.
Let $\Cs=\{C_1,C_2,\ldots, C_k \}$ be the set of all cycles in $G$. 
A \emph{cycle pattern} is a function $\psi \colon \Cs \to \{+,0,-\}$. 
We assume $G$ is \emph{weighted}, i.e., there is a function $w \colon E \to \RR$. We will often interpret the weight function as a vector $w \in \RR^E$. One can define the \emph{weight} of a cycle $C$ as the sum of the weights of the edges of the cycle: $w(C)=\sum_{e\in C\cap E}w(e)$. By only considering if the weight of a cycle is positive, zero or negative, the weight function gives rise to a cycle pattern $\psi_w$. 
In this case, we say that the cycle pattern is \emph{induced} by $w$. We also say that such a cycle pattern is \emph{realizable} with \emph{realization} $w$. 
A realization is \emph{integral} if its values are only integers.

We can describe the realizations of a cycle patterns as follows:

\begin{definition} \label{def:realization+cone}
    Given a realizable cycle pattern $\psi$, the set of its realizations forms the \emph{realization cone} in $\RR^{E}$ defined by the following inequalities:
\begin{alignat}{2}
    &\sum_{e\in C} w(e) > 0 &\quad& \forall C\in \Cs: \psi(C)=+\nonumber \, ,\\
    &\sum_{e\in C} w(e) < 0 && \forall C\in \Cs: \psi(C)=-\label{eq:cycle_cone} \, ,\\
    &\sum_{e\in C} w(e) = 0 && \forall C\in \Cs: \psi(C)=0 \, .\nonumber
\end{alignat}
Note that this cone is a set of the form $\{w \in \RR^{E} \colon Aw > 0, Bw = 0\}$, where $A,B$ are matrices with $|E|$ columns. 
The rows of $A$ are of the form $\chi(C)$ for all $C \in \Cs$ such that $\psi(C) = +$ and of the form $-\chi(C)$ for all $C \in \Cs$ such that $\psi(C) = -$.
The rows of $B$ are of the form $\chi(C)$ for all $C \in \Cs$ such that $\psi(C) = 0$.
\end{definition}

Finally, we also distinguish another class called \emph{parity-realizable} cycle patterns, which we will later relate to parity games. 
Whenever considering parity-realizability, to avoid confusion with normal cycle patterns, we refer to the function $w$ as \emph{priorities} of the edges, instead of weights. 
This follows the usual convention for parity games.
Given a priority function $w:E\to \mathbb{Z}$, we get the cycle pattern $\psi^w$, with $\psi^w(C)=+$ if the largest priority of $C$ is even, and $\psi^w(C)=-$ if its largest priority is odd.  
We say in that case that $\psi^w$ is \emph{parity-induced} by $w$, and that the cycle patterns we can obtain in this way are \emph{parity-realizable}. 
Every parity-realizable cycle pattern is also realizable: if $w$ is a priority function that parity-realizes $\psi$, then the weight function $w'$ given by $w'(e)=(-|V|)^{w(e)}$ gives a realization of $\psi$.

\subsubsection{The cycle space}\label{sec:cyclespace}

As a technical tool, we include a statement about cycle spaces. 
Denote the set of incoming and outgoing edges of a vertex $v$ by $v_{in}$ and $v_{out}$, respectively. We consider the transpose of the incidence matrix $M \in \RR^{E \times V}$ given by 
    \begin{equation}\label{eq:potential_matrix}
        M_{e,v} = \begin{cases}
        1 &\text{if $e \in v_{out}$},\\
        -1 &\text{if $e \in v_{in}$},\\
        0 &\text{otherwise}.
        \end{cases}
    \end{equation}
Note that the matrix $M$ has two nonzero entries per row, one equal to $1$ and one equal to $-1$. In literature, there exists a more relaxed notion of directed cycles which, among other things, allows for using directed edges in the `wrong' direction, at the cost of reversing its sign in the characteristic vector (see e.g. \cite[Ch. 14.2]{godsil_algebraic_2001}, \cite{KAVITHA2009199}). It is well-known that for that definition of cycle, the space spanned by the (characteristic vectors of) cycles of the graph is given by $\ker(M^T)$. If the graph is connected, then the dimension of $\ker(M^T)$ is $m-n+1$. This space is called \emph{flow space} or \emph{cycle space} of the graph . Since we only allow for cycles that follow the edge directions, our cycles will not always span the usual cycle space for arbitrary digraphs. This can happen for example when an edge is contained in an undirected cycle, but not in any directed cycles. However, in the following we show that, in the case of strongly connected graphs, the span of the cycles of $\Cs$ coincides with the common notion of cycle space.

\begin{lemma}\label{lem:dicyclebasis}
We have $\operatorname{span}\setOf{\chi(C)}{C\in \Cs} = \ker(M^T)$. In particular, \\ $\dim(\operatorname{span}\setOf{\chi(C)}{C\in \Cs})=m-n+1$.
\end{lemma}
\begin{proof}
We want to show that $\operatorname{span}\setOf{\chi(C)}{C\in \Cs}$ is precisely the set of all $y \in \RR^{E}$ such that $M^Ty=0$, which means they satisfy the equalities
    \[
    \forall v \in V, \ \sum_{e\in v_{out}}y_e = \sum_{e'\in v_{in}}y_{e'} \, .
    \]
    If $C \in \Cs$, then the vector $\chi(C)$ satisfies these equalities (both sides are equal to $0$ if $v$ does not belong to $C$ and to $1$ if $v$ belongs to $C$), hence $\operatorname{span}\setOf{\chi(C)}{C\in \Cs} \subseteq \ker(M^T)$.  Conversely, suppose that $y \in \RR^{E}$ satisfies these equalities. We claim that we can get the zero vector by adding and subtracting multiples of the vectors of the form $\chi(C)$ from $y$. We do so by the following procedure: first, add multiples of the vectors $\chi(C)$ to make sure that all elements of $y$ are nonnegative (we can do so since every edge is contained in some cycle, because the graph is strongly connected). Next, we repeat the following: pick any positive element of $y$, say with index $(i,j)$. Since the totals of $y$ for the incoming and outgoing edges of vertex $j$ are the same, there must be an outgoing edge $(j,k)$ with $y_{(j,k)}>0$. In the same way, $k$ must have a positive outgoing edge, so continuing this we find a cycle $C$ on which $y$ is positive. Now we subtract $\left(\min \setOf{y_e}{e\in C}\right)\chi(C)$ from $y$ to get a nonnegative vector with at least one more element equal to zero. So the procedure must terminate with the zero vector, and we conclude that $y\in \operatorname{span}\setOf{\chi(C)}{C\in \Cs}$. We conclude that indeed $\operatorname{span}\setOf{\chi(C)}{C\in \Cs}= \ker(M^T)$. Therefore its dimension equals that of the usual cycle space/flow space, which for a connected graph is $m-n+1$.
\end{proof}

\subsection{Realizability of cycle patterns}\label{sec:realizability}
We discuss some necessary and sufficient conditions for a cycle pattern to be realizable. 
It may not be surprising that not every cycle pattern is realizable as the following example demonstrates. 
   
   \begin{figure}[htb]
        \centering
        \includegraphics[width=0.3\linewidth]{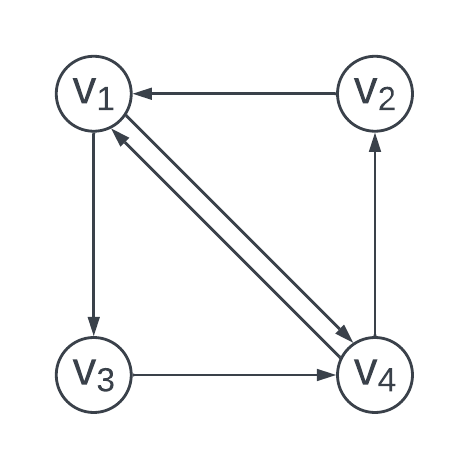}
        \caption{Digraph with a cycle pattern that is not realizable.}
        \label{fig:MPG_impossible}
    \end{figure}

\begin{example}
Suppose that we have some cycle pattern with $\psi(v_1,v_4)=\psi(v_1,v_3,v_4,v_2)=+$ and $\psi(v_1, v_4, v_2)=\psi(v_1, v_3, v_4)=-$ on the graph depicted in \cref{fig:MPG_impossible}. 
It is not hard to see that it is impossible to find a realization for this cycle pattern: 
we have two positive and two negative cycles, and both pairs of cycles contain exactly all the edges of the graph. 
So in a realization, the total sum of all edge weights would need to be both positive and negative at the same time, a contradiction.    
\end{example}

The construction of this example already highlights the crucial obstructions for realizability. We show that we can characterize realizability by existence of certain sequences of positive and negative cycles that contain the same edges. We now make this more formal.

 \begin{definition}
     Given a cycle pattern $\psi$, we call two nonempty sequences of cycles $C_1,\dots,C_p$ and $D_1,\dots,D_q$
 an \emph{opposing pair} if 
 \begin{equation}\label{eq:signed-circuit-sequences}
        P = \{\psi(C_1),\dots,\psi(C_p)\}\subseteq \{0,+\}, \ Q = \{\psi(D_1),\dots,\psi(D_q)\}\subseteq \{0,-\}, \ P \cup Q \neq \{0\} \,.    
\end{equation}
 \end{definition}
Now we are ready to characterize realizability with the following theorem:

\begin{theorem}\label{thm:char1}
    The following statements are equivalent:
    \begin{enumerate}
        \item The cycle pattern $\psi$ is realizable. 
        \item The cycle pattern $\psi$ has an integral realization $w$ with 
        $\max_{e\in E}|w(e)|\leq (n+1)^{(m - n + 1)/2}$.
        \item There does not exist an opposing pair $C_1,\dots,C_p$ and $D_1,\dots,D_q$ with
        \begin{equation}\label{eq:equal_edges}
        \sum_{i=1}^{p} \chi(C_i) = \sum_{j=1}^{q} \chi(D_j) \enspace .
        \end{equation}     
    \end{enumerate}
\end{theorem}
For example, in Figure \ref{fig:MPG_impossible}, the third condition would be falsified by taking the two positive cycles and the two negative cycles as an opposing pair $C_1,C_2,D_1,D_2$. We call such an opposing pair a \emph{non-realizability witness}. Now we prove the theorem:

\begin{proof}
    ($1 \implies 3$) Suppose for contradiction that $\psi$ is realizable with realization $w$, and that we have an opposing pair as in the third statement. 
    Since $P\cup Q\neq \{0\}$, there must be a \break $+$-cycle in $P$ or a $-$-cycle in $Q$. 
    Assume without loss of generality that $P$ has a $+$-cycle. Since  the sequence $C_1, \ldots, C_p$ only consists of $+$-cycles and $0$-cycles, we have $w(C_i) = \sum_{e\in C_i}w(e)\geq 0$ for every cycle~$C_i$. 
    Hence $\sum_{i=1}^{p} w(C_i) >0$. Similarly, since $D_1, \ldots, D_q$ consists only of $0$-cycles and $-$-cycles, we get $\sum_{j=1}^{q} w(D_j) \leq 0$.
    Finally, taking the scalar product of both sides of \cref{eq:equal_edges} with the weight vector $w$ gives
    \[
    \sum_{i=1}^{p} w(C_i) = \sum_{j=1}^{q} w(D_j) \, ,
    \]
    which yields a contradiction.  

    \smallskip
    
    ($3\implies 1$) To get a contradiction, we suppose there is no realization of $\psi$. 
    We will derive that there must exist an opposing pair $C_1, \ldots, C_p$ and $D_1,\ldots, D_q$ satisfying \cref{eq:equal_edges}. 
    
    If there is no realization of $\psi$, that means that the realization cone $\{w \in \RR^{E} \colon Aw > 0, Bw = 0\}$ defined in \cref{eq:cycle_cone} is empty. 
    Since the cone is empty, the smaller cone composed of only rational vectors $\{w \in \QQ^{E} \colon Aw > 0, Bw = 0\}$ is also empty. Furthermore, note that the matrix $A$ has at least one row because the zero vector is not in the realization cone, and note that the matrices $A,B$ have integer coefficients.
    Therefore, a version of Farkas' lemma with strict inequalities (sometimes called a ``transposition theorem'') given, e.g., in \cite[Section~1.6]{StoerWitzgall:1970} states that the cone \break $\{w \in \QQ^{E} \colon Aw > 0, Bw = 0\}$ is empty if and only if there exist rational vectors $y,z$ such that $y \ge 0$, $y \neq 0$, and $y^TA + z^TB = 0$. Recall that the rows of $A$ and $B$ each correspond to a cycle of $\Cs$, so we can likewise say that $(y,z)$ belongs to $\QQ^{\Cs}$. We can then index the entries of $y,z$ by the corresponding cycles. Furthermore, by scaling, we can suppose that the entries of the vectors $y,z$ are integers. We now construct two sequences of cycles $C_1,C_2, \ldots, C_p$ and $D_1, D_2, \ldots, D_q$ as follows:
\begin{itemize}
    \item Every cycle $C\in \psi^{-1}(+)$ occurs $y_C$ times in the sequence $C_1,\ldots, C_p$.
    \item Every cycle $C\in \psi^{-1}(-)$ occurs $y_C$ times in the sequence $D_1,\ldots, D_p$.
    \item Every cycle $C\in \psi^{-1}(0)$ occurs $z_C$ times in the sequence $C_1,\ldots, C_p$ if $z_C > 0$ and $|z_C|$ times in the sequence $D_1,\ldots, D_q$ if $z_C < 0$.
\end{itemize}
Let $P = \{\psi(C_1),\dots,\psi(C_p)\}\subseteq \{0,+\}$ and $Q = \{\psi(D_1),\dots,\psi(D_q)\}\subseteq \{0,-\}$. Then, the fact that $y \neq 0$ implies $P \cup Q \neq \{0\}$. Furthermore, the equality $y^TA + z^TB = 0$ gives \break $\sum_{i=1}^{p} \chi(C_i) = \sum_{j=1}^{q} \chi(D_j)$, and thus we have our non-realizability witness.

    \smallskip

    ($1\implies 2$) 
    We consider again the realization cone. 
    If this cone is nonempty, then, by suitable scaling, the polyhedron 
    \begin{equation} \label{eq:realization-polyhedron}
        \mathcal{P} = \{w \in \RR^{E} \colon Aw \ge \bm{1}, Bw = 0\}    
    \end{equation}
     is also nonempty, where we denote $\bm{1} = (1,1,\dots,1)$. 
     Following the proof of \cite[Thm 10.1]{Schrijver:1987}, we show that this cone has an integer point that satisfies the bound of our theorem. 
     Take any minimal face $F$ of $\mathcal{P}$. 
     This face is defined by $A'x=b'$, where $A'$ consists of a number of linearly independent rows of $\begin{bmatrix}
        A\\
        B
    \end{bmatrix}$ and $b'$ consists of the related 1's and 0's on the right side of the (in)equalities of $\mathcal{P}$. Since $\operatorname{im}((A')^T)\subseteq \operatorname{span}\setOf{\chi(C)}{C\in \Cs}$, it follows from Lemma \ref{lem:dicyclebasis} that $\operatorname{rank}(A')\leq m-n+1$. Since $A'$ has full rank, we may rearrange its columns so that $A'=\begin{bmatrix}
        A'' & A'''
    \end{bmatrix}$, where $A''$ is invertible. Then the point $w^*=\begin{bmatrix}
        (A'')^{-1}b'\\
        0
    \end{bmatrix}$ is in $F$. Cramer's rule tells us that for the $e$-row of $A''$, we get $w^*_e=\det(A''_e)/\det(A'')$, where $A''_e$ is obtained from $A''$ by replacing the $e$-column by $b'$.

    Since $w^* \in F \subseteq \mathcal{P}$, the scaled vector $\hat{w} = \det(A'')w^*$ also belongs to the cone of all realizations $\{w \in \RR^{E} \colon Aw > 0, Bw = 0\}$ and its coordinates are integers given by $\hat{w}_e = \det(A''_e)$ for edges that have an associated column of $A''$, and $\hat{w}_e=0$ for all other edges. We can bound these by Hadamard's inequality, we have $|\hat{w}_e| \le \prod_{C}\|u_C\|_2$, where $u_C$ is the row of $A''_e$ corresponding to cycle $C$. We have $\|u_C\|_2 \le \sqrt{n+1}$ for any of the (at most) $m - n + 1$ rows of $A''_e$: that is because any cycle has at most $n$ edges, so there are at most $n+1$ entries of $u_C$ equal to $-1$ or $1$. Hence, $|\hat{w}_e| \le (n+1)^{(m-n+1)/2}$ for every $e$.\\
    ($2\implies 1$) Trivial.
    \end{proof} 

    The proof of the theorem tells us that solving the feasibility problem for the linear inequality system in \eqref{eq:realization-polyhedron} actually allows one to compute a realization of a cycle pattern, if it exists. However, the linear program may be very large.

\subsection{About the size of non-realizability witnesses}
\label{sec:witnesssize}

    We can derive some bounds on the size of a non-realizability witness with similar techniques to the proof of Theorem \ref{thm:char1}.

\begin{proposition}
    Given a non-realizable cycle pattern $\psi$, there is a witness with at most $m-n+2$ distinct cycles. 
\end{proposition}
\begin{proof}
    Looking carefully at the proof step $(3\implies 1)$, we can replace the condition $y\neq 0$ by $\bm{1}^{T}y\geq 1$, since we want $y$ to be integer. 
    Given that cyle pattern $\psi$ is not realizable, the inequalities $y\geq 0, \bm{1}^Ty\geq 1,y^TA+z^TB=0$ define a nonempty polyhedron. Take any facet of this polyhedron, it is defined by some subset of the rows of the following equation:
    \begin{equation}\label{eq:witnesscone}
        \begin{bmatrix}
        A^T & B^T\\
        I & 0\\
        \mathbf{1}^T & 0\\
    \end{bmatrix}\begin{bmatrix}
        y\\z
    \end{bmatrix}=\begin{bmatrix}
        0\\
        0\\
        1
    \end{bmatrix}
    \end{equation}
    Say the facet is defined by $A'\begin{bmatrix}
        y\\z
    \end{bmatrix}=b'$. Then like before we can say $A'=\begin{bmatrix}
        A'' & A'''
    \end{bmatrix}$, with $A''$ invertible, yielding the rational point $(y,z)=((A'')^{-1}b', 0)$ in the polyhedron (note: the right hand side is denoted in this way for convenience of notation, actual order of the elements may be different). However, looking at the matrix from \cref{eq:witnesscone}, at most $m-n+2$ linearly independent rows are not unit vectors. It follows that at most $m-n+2$ rows of $A''$ are not unit vectors. Any row of $A''$ that is a unit vector dictates that one element of $(A'')^{-1}b'$ is $0$, and therefore $(A'')^{-1}b'$ can have at most $m-n+2$ nonzero elements. If we then scale $(A'')^{-1}b'$ to an integer vector, this gives us a non-realizability witness similar to the proof of Theorem \ref{thm:char1}, with at most $m-n+2$ distinct cycles. That concludes the proof. 
\end{proof}
    
    However, this does not imply that there always exist non-realizability witnesses with \break $p+q\leq m-n+2$. It is not even clear that its size can be bounded by a polynomial in $m$ and $n$. 
    We give an exemplary sketch that it may actually be exponential in $m$ and $n$.

    \begin{figure}[H]
        \centering
        \includegraphics[width=\linewidth]{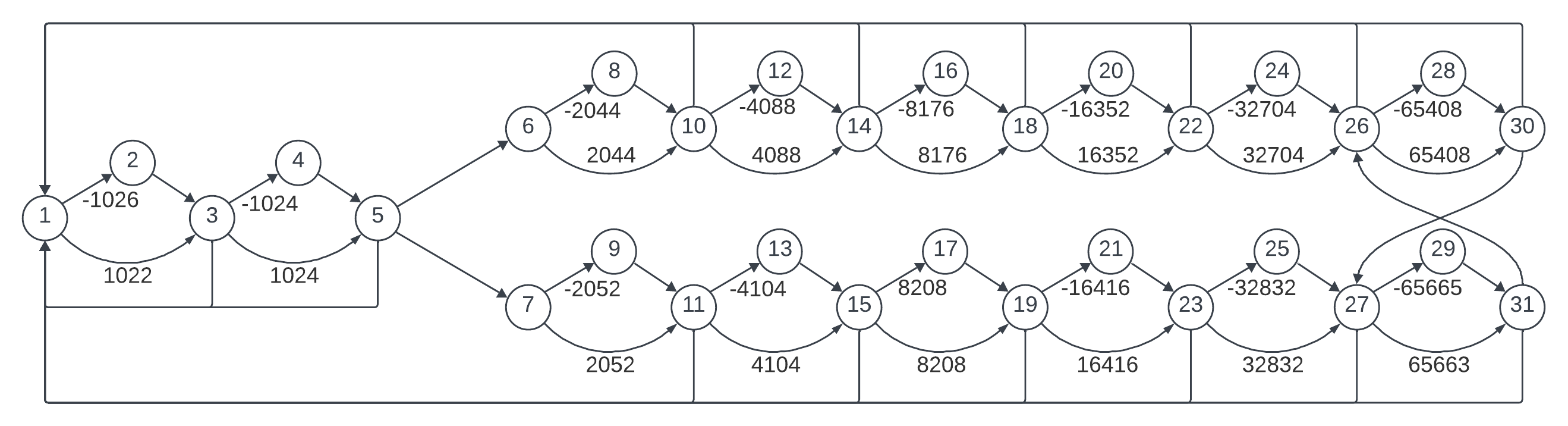}
        \caption{Graph that has a cycle pattern for which only large non-realizability witnesses exist. Edge weights are shown on all the edges that have nonzero weight.}
        \label{fig:farkas_limit}
    \end{figure}
    Take the graph in \cref{fig:farkas_limit}. Let $\psi_w$ be the cycle pattern induced by its weights. There are no 0-cycles w.r.t. $\psi_w$. Let $\psi$ be the cycle pattern that is the same as $\psi_w$ for almost all cycles, except that the sign is reversed for the following cycles:
    \begin{table}[H]
        \centering
        \begin{tabular}{c|c|c}
           $C$  &  $\psi_w(C)$ & $\psi(C)$\\ \hline
           $(27,31,26,28,30)$  & $+$ & $-$\\
            $(27,29,31,26,30)$ & $-$ & $+$\\
            $(1,3,5,7,11,15,19,23,25,27,31,26,28,30)$ & $+$ & $-$\\
            $(1,2,3,4,5,7,9,11,13,15,17,19,21,23,27,29,31,26,30)$ & $-$ & $+$\\
            $(1,3,5,6,10,14,18,22,24,26,30,27, 29,31)$ & $-$ & $+$\\
            $(1,2,3,4,5,6,8,10,12,14,16,18,20,22,26,28,30,27,31)$ & $+$ & $-$
        \end{tabular}
    \end{table}
    There are no 0-cycles in this cycle pattern. Therefore we can simply use an integer linear program to find a non-realizability witness with $p+q$ minimized. Solving this tells us that $\psi$ is not realizable, and the smallest witness has size $p+q=64$ \footnote{The graph shown in Figure \ref{fig:farkas_limit} is part of a family of graphs $(G_n)_{n\in\mathbb{N}}$. We claim that the smallest non-realizability witness for $G_n$ has $p+q=2^n$ for $n\geq 2$. This is confirmed experimentally for $n\leq 10$, see \url{https://github.com/MatthewMaat/Cycle_Patterns}.}, contrary to $m-n+2=60-31+2=31$.

\subsection{Parity-realizability}

As mentioned before, every cycle pattern that is parity-realizable is also realizable. 
However, the converse is not true, as shown in the following construction.

\begin{theorem} \label{thm:non-parity-realizable}
    There are realizable cycle patterns that are not parity-realizable.
\end{theorem}
\begin{proof}
    Consider the complete directed graph on $4$ vertices, with weights given as in \cref{fig:MPGpattern}. In this example, the edges between two pairs of vertices get weight $3$, and the other edges get weight $-2$. We argue that every edge is contained in a cycle with positive weight and in a cycle with negative weight. Because of symmetry, there are only two types of edges we need to check this for.
    Firstly, the edge $(v_1,v_2)$ is part of the cycle $v_1,v_2, v_4$ with weight $-1$ and of the cycle $v_1, v_2, v_3, v_4$ with weight $2$. 
    Secondly, $(v_1, v_4)$ is in cycle $v_1, v_4, v_2$ with weight $-1$ and in cycle $v_1, v_4, v_2, v_3$ with weight $2$. So all edges are in a positive and negative cycle. Now we argue that this implies that $\psi_w$ is not parity-realizable. Suppose there is a parity realization $w'$ of $\psi_w$, and suppose $w'(e)$ is the highest priority of $w'$. If $w'(e)$ is even, then every cycle $C$ containing $e$ must have $\psi_w(C)=+$, which is false. Likewise, if $w'(e)$ is odd, then then every cycle $C$ containing $e$ must have $\psi_w(C)=-$, which is also not true. We conclude that such a priority function $w'$ cannot exist.
\end{proof}

\begin{figure}[htb]
        \centering
        \includegraphics[width=0.4\linewidth]{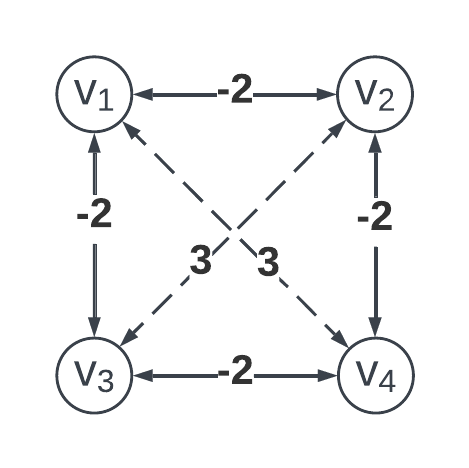}
        \caption{Complete graph on 4 vertices showing the example for \cref{thm:non-parity-realizable}}
        \label{fig:MPGpattern}
    \end{figure}

We see that, in a parity-realizable cycle pattern, there must always be an edge such that all cycles containing it have the same sign. Even more, if some set $S$ of edges is deleted from the graph, there will still be an edge with the highest priority in the parity-realization. This edge then must still have the same property that all cycles containing it have the same sign. This idea, in fact, precisely characterizes which cycle patterns are parity-realizable. This is formalized in the following theorem.

\begin{theorem}\label{thm:charPG}
    A cycle pattern $\psi$ is parity-realizable if and only if the following two conditions hold:
    \begin{enumerate}[label=(\roman*)]
        \item $\psi$ has no $0$-cycles.
        \item For all nonempty $S\subseteq E$, there exists an edge $e\in S$ such that all cycles $C$ with $e \in C$ and $C \subseteq S$ have the same sign. 
    \end{enumerate}
\end{theorem}
\begin{proof}
    For the first direction, suppose we have a cycle pattern $\psi^w$ parity-induced by priority function $w$. Clearly it has no zero cycles. Let $S$ be a nonempty subset of edges. If there are no cycles in $S$ then the second statement trivially holds. So suppose there are cycles $C\in S$. Let $\mathcal{C}_S$ be the union of these cycles, and let $e$ be an edge with the highest priority in $S$. Obviously, if $w(e)$ is even, every cycle in $\mathcal{C}_S$ containing $e$ has sign $+$, and they have sign $-$ if the priority of $e$ is odd. So in all cases the second condition holds, and this completes the first direction of the proof.

    On the other hand, suppose we have a cycle pattern $\psi$ that satisfies the condition. We then want to construct a parity realization of $\psi$. First, take $S=E$, we know that there is an edge $e\in S$ such that every cycle containing $E$ has the same sign. We call this edge $e_1$, and assign it the highest priority of the parity game: we make sure that $w(e_1)$ is even if all cycles containing it are $+$, and otherwise $w(e_1)$ will be odd (if no cycles contain $e_1$, we may pick its parity arbitrarily). Independent of what other priorities we choose, as long they are smaller than $w(e_1)$, we have $\psi^w(C)=\psi(C)$ for all cycles containing $e_1$.
    
    Now pick $S=E\backslash\{e_1\}$, and let $e_2$ be an edge such that all cycles in $S$ containing $e_2$ have the same sign. We now make sure that $w(e_2)$ is the second highest priority, with again a similar choice for its parity. With the same argument, $\psi(C)=\psi^{w}(C)$ in every cycle containing $e_1$ or $e_2$. We then remove $e_2$ from $S$ and repeat this procedure, until we eventually fixed $w$ for all edges. We then have $\psi^w(C)=\psi(C)$ for all cycles $C$. We conclude that indeed $\psi$ is parity-realizable.
\end{proof}

The previous statement shows that parity-realizable cycle patterns have a simpler structure. 
We show in \cref{sec:extremal-cycle-patterns} that nevertheless there are parity-realizable cycle patterns which still need big weights in any representation. 

We conclude this section with a question asking to quantify how rare parity-realizable cycle patterns are among all realizable cycle patterns. 

\begin{question}
    Which fraction of realizable cycle patterns is parity realizable? 
\end{question}


\subsection{Lower bounds on cycle patterns}
\label{sec:extremal-cycle-patterns}

We complement the upper bound on the largest size of the weights of a realization from \cref{thm:char1} by a lower bound construction. 

\begin{theorem} \label{thm:expweights}
    There are families of graphs with realizable cycle patterns for which every integer realization has an edge $e$ with $|w(e)|= 2^{\Omega(m)}$. Similar results hold for special cases:
    \begin{itemize}
        \item For parity-realizable cycle patterns, every realization has an edge with $|w(e)|=2^{\Omega(m)}$.
        \item For simple graphs there is an edge with $|w(e)|=2^{\Omega(n^2)}$.
        \item For simple graphs with bounded feedback arc set number\footnote{The feedback arc set number is the smallest possible cardinality of a set $S\subseteq E$ such that the digraph $(V,E\backslash S)$ is acyclic.}, there is an edge with $|w(e)|=2^{\Omega(n)}$.
    \end{itemize}
\end{theorem}

\begin{proof}
Consider the graph $G_i$ with priorities $\hat{w}$ on the left of \cref{fig:MPG_minweight}. 
There are $8i$ edges in this graph, with $\hat{w}(e_j)=j$ for $j=1,2,\ldots, 8i$. 
There are four vertices $v_1,v_2,v_3$ and $v_4$, where for each~$k$, the vertex $v_k$ has the outgoing edges $e_{8j+2k-9}$ and $e_{8j+2k-8}$ for $j=1,2,\ldots, i$. 
The priority function $\hat{w}$ parity-induces the cycle pattern $\psi^{\hat{w}}$ on this graph. 
For example, cycle $C=(e_1,e_4,e_{13},e_8)$ has priorities $1,4,13,8$, of which the largest is odd, so $\psi^{\hat{w}}(C)=-$. 

Now let $w \in \ZZ^E$ be an arbitrary integer realization of $\psi^{\hat{w}}$ as shown on the right of \cref{fig:MPG_minweight}. 

    \begin{figure}[htb]
        \centering
        \includegraphics[width=\linewidth]{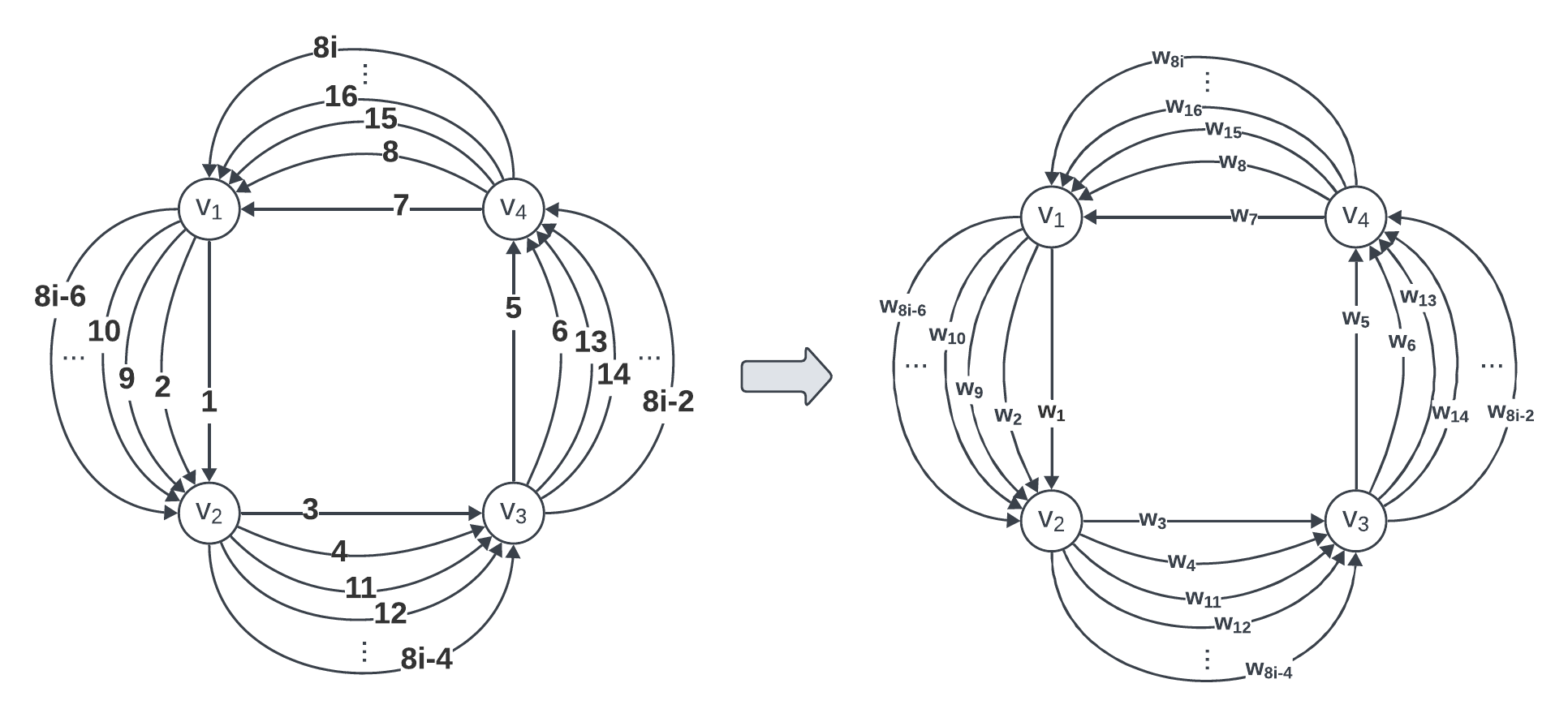}
        \caption{Left: graph $G_i$ with priorities $\hat{w}$ on the edges. Right: if we use weights $w$ on $G_i$, the induced cycle pattern needs to equal $\psi^{\hat{w}}$.}
        \label{fig:MPG_minweight}
    \end{figure}
    
We prove by induction on $j$ that $w_{2j}-w_{2j-1}> F_{j-3}$ for all $j\geq 4$, where $F_p$ is the $p$-th Fibonacci number. 
This implies that either $|w_{8i}|$ or $|w_{8i-1}|$ has weight exponential in $m=8i$.

For the induction basis, note that $\psi^{\hat{w}}(e_1,e_3,e_5,e_7)=-$ and $\psi^{\hat{w}}(e_1,e_3,e_5,e_8)=+$, implying $w_1+w_3+w_5+w_7<0$ and $w_1+w_3+w_5+w_8>0$. 
As we have an integer realization, this yields $w_1+w_3+w_5+w_8\geq 1$. 
Combining both equations gives $w_8-w_7>1=F_1$. 
Analogously, one sees $w_{10}-w_9>1=F_2$.

Now assume $w_{2j-4}-w_{2j-5}>F_{j-5}$ and $w_{2j-2}-w_{2j-3}>F_{j-4}$ for some $j\geq 6$. 
By construction, we have $\psi^{\hat{w}}(e_{2j-6},e_{2j-4},e_{2j-2},e_{2j-1})=-$, resulting in $w_{2j-6}+w_{2j-4}+w_{2j-2}+w_{2j-1}<0$, and likewise $w_{2j-6}+w_{2j-5}+w_{2j-3}+w_{2j}>0$. 
Combining these gives
\begin{align*}
    0&<(w_{2j-6}+w_{2j-5}+w_{2j-3}+w_{2j})-(w_{2j-6}+w_{2j-4}+w_{2j-2}+w_{2j-1}) \\
    &\stackrel{*}{<} (w_{2j}-w_{2j-1})-F_{j-4}-F_{j-5} \, ,
    \end{align*}
where $*$ follows from the induction hypothesis. 
We get $w_{2j}-w_{2j-1}>F_{j-4}+F_{j-5}=F_{j-3}$ completing the proof for general graphs.

Next, we address the three special cases from the theorem. 
First of all, the construction from \cref{fig:MPG_minweight} gives a parity-realizable cycle pattern. To get a simple graph, we need a slight modification of the graphs $G_i$.
Instead of a cycle with four vertices, we make one with $k$ vertices $v_1,v_2,\ldots, v_k$, where the vertex $v_p$ has outgoing edges of priorities $2jk+2p-1$ and $2jk+2p$ for $j=0,1,\ldots, i-1$. 
To make the graph simple, we suppose further that $k$ is divisible by $3$. 
We add vertices $v_{a,b}$ for $a=1,2,\ldots, 2i$ and $b=1,2,3$ to the graph. 
Then, for any $p$ we split the outgoing edges of $v_p$ as follows. 
Suppose that the edge $(v_p, v_q)$ has the $j$-th smallest priority among the outgoing edges of $v_p$. 
We replace the edge $(v_p, v_q)$ by edges $(v_p, v_{j,b})$ and $(v_{j,b},v_q)$, such that $b = p\mod{3}$. 
We assign the same priority as the original edge to the new edges. To finish the construction, we add the edges $(v_1,v_2),\dots,(v_{k-1},v_k),(v_k,v_1)$ to the graph, all with priority $0$. 
We denote this graph by $\widehat{G} = (\widehat{V},\widehat{E})$. 

The graph $\widehat{G}$ is simple. 
Indeed, every outgoing edge of $v_p$ goes to a different vertex by definition. 
Furthermore, every outgoing edge of $v_{j,b}$ goes to a different vertex because an edge of the form $(v_{j,b},v_q)$ arises in a unique way, namely by taking $v_p$ to be the vertex that precedes $v_q$ in the original cycle and splitting the edge $(v_p,v_q)$ of the $j$-th smallest priority. 
Let $w \in \ZZ^{\widehat{E}}$ be an integer realization of the cycle pattern on $\widehat{G}$. 
For every $\ell \in \{1,\dots,2ki\}$, let $u_{\ell} = w_e + w_e'$ be the sum of weights of the two edges with priority $\ell$. 
The proof follows by analyzing subgraphs of the form depicted in \cref{fig:MPG_minweight_simple}. 
In this subgraph, the cycle going through edges with priorities $13,15,18$ has positive weight, while the cycle going through edges with priorities $14,16,17$ has negative weight. 
This implies that $u_{13} + u_{15}+ u_{18} > u_{14} + u_{16} + u_{17}$, or equivalently $u_{18} - u_{17} > (u_{16} - u_{15}) + (u_{14} - u_{13})$. 
Since such a subgraph exists in $\widehat{G}$ for any sequence of six consecutive priorities that starts with an odd priority, we get $u_{2\ell+4} - u_{2\ell+3} > (u_{2\ell+2} - u_{2\ell+1}) + (u_{2\ell} - u_{2\ell-1})$ for all $\ell \in \{1,\dots,ki-2\}$. 
Moreover, we have $u_2 - u_1 \ge 1$ by analyzing the cycles that go only through the edges with priorities $0,1,2$. 
Analogously, $u_4 - u_3 \ge 1$. 
Hence $u_{2\ell} - u_{2\ell-1} \ge F_{\ell}$ for all $\ell$ and, in particular, $u_{2ki} - u_{2ki-2} \ge F_{ki}$. 
As a consequence, at least one of the four edges with priorities $2ki, 2ki-1$ satisfies $|w_e| \ge 2^{cki}$ for some absolute constant $c > 0$ that does not depend on $k,i$. 
By taking $i = k$ we get a graph $\widehat{G}$ with $k + 6i = 7k$ vertices. 
It follows then that there is an edge with $|w(e)|=2^{\Omega(n^2)}$.

\begin{figure}[htb]
    \centering
    \includegraphics[width=0.8\linewidth]{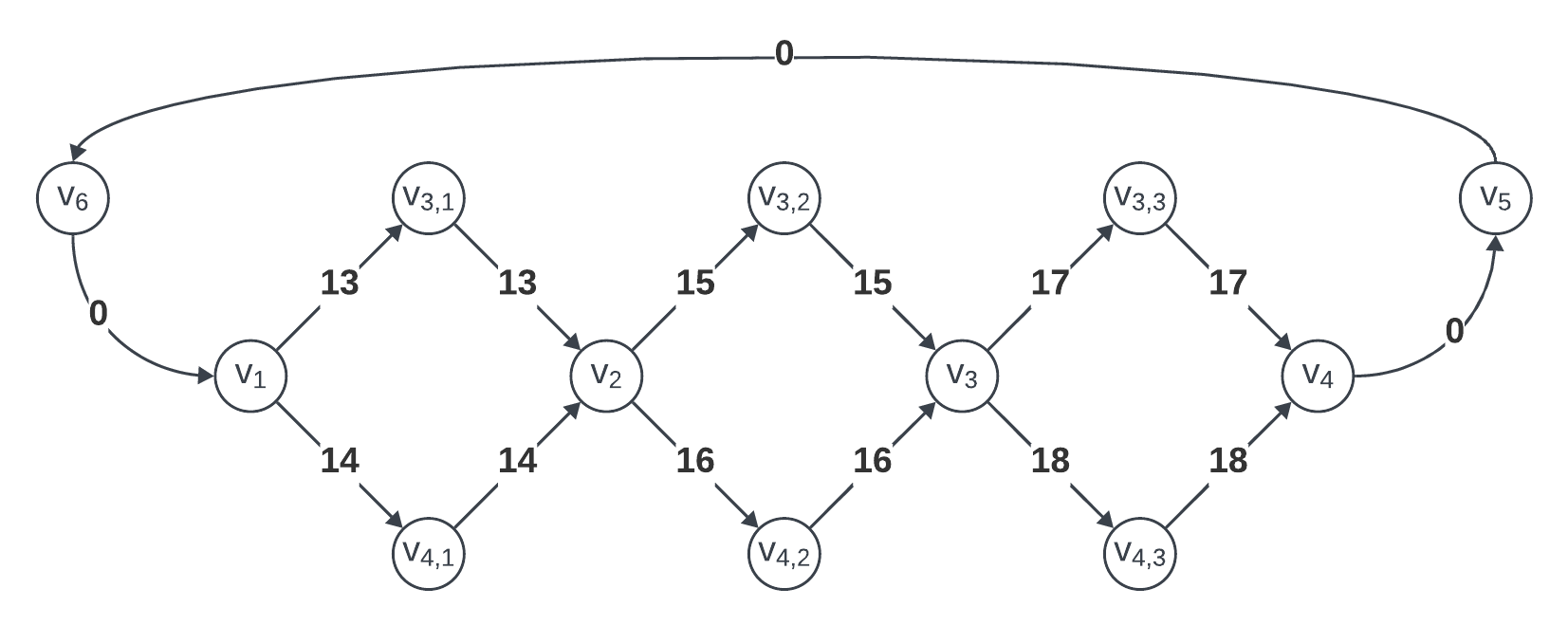}
    \caption{A subgraph of $\widehat{G}$ in the case $k=6,i=3$.}
    \label{fig:MPG_minweight_simple}
\end{figure}

In this construction, there are many intertwined cycles, so this graph will generally have a large feedback arc set number. 
Hence we turn the graph from \cref{fig:MPG_minweight} into a simple graph with feedback arc set number $1$ with a different method: first, we make it simple by adding an internal node to every edge, splitting the edge into two edges with the same priority. 
The new graph has a linear number of edges in terms of $n$. 
Then, we split $v_1$ into nodes $v_1'$ and $v_1''$. 
For every former edge of the form $(x,v_1)$ we add the edge $(x,v_1')$ with the same priority. 
Likewise, for every former edge $(v_1,x)$ we add the edge $(v_1'',x)$. 
We also add the edge $(v_1',v_1'')$ with priority 0. 
Now removing the last mentioned edge would turn the graph into an acyclic graph, hence the new graph has feedback arc set number 1. 
Moreover, this modification has no effect on the proof that the realization has large weights, hence for every realization $w$ there must be an edge $e$ with $|w(e)|=2^{\Omega(n)}$, and this completes the proof.
\end{proof}

The argument extends to various other graph complexity measures like DAG-width, Kelly-width or directed treewidth, as long as they are at most 2. 
This is because these measures are all bounded by the feedback arc set number $+1$ (see for example \cite[p. 62]{rehs2022comparing}).

\begin{remark}
Recall from \cref{thm:char1} the upper bound of $(n+1)^{(m-n+1)/2}$ for the maximal weight that occurs in a smallest realization. 
Note that this behaves asymptotically like $2^{\Theta(m\log(n))}$, while \cref{thm:expweights} shows a lower bound of $2^{\Omega(m)}$. 
This is not tight and we conjecture a lower bound of $2^{\Omega(m\log(n))}$. 

Recall that we derived the upper bound by solving a system of equations deduced from~\cref{eq:realization-polyhedron}. 
Specifically, $\hat{w}=\det (A'')\begin{bmatrix}
        (A'')^{-1}b'\\
        0
    \end{bmatrix}$ as defined in the proof is an integer point.
In particular, since the vector $b'$ of this system contains only $-1,0$ or $1$, if we want `big' solutions to the system, this requires either a `big' inverse of $A''$, or a big determinant of $A''$, and enough asymmetry in $(A'')^{-1}$ to prevent a smaller solution from existing. 

Since $A''$ can be written as a $0$-$1$-matrix, this is closely related to the fundamental question of determining the maximum possible absolute value of an entry of the inverse of an invertible $n\times n$-matrix with 0, 1 entries. 
Denoting by $\chi_1(n)$ this maximum, it was shown in \cite{alon1997anti} that $\chi_1(n) = n^{(\frac{1}{2}+o(1))n}$. 
This bound had several direct implications, e.g., for the degrees of regular multi-hypergraphs and for the flatness of simplices. 
The matrices arising in the construction from \cite{alon1997anti} are rather dense in the sense that they have many non-zero entries. 
The matrices in the proof of \cref{thm:char1} are \emph{sparse} in the sense that that they are incidence matrices of cycles, so they have $O(n)$ non-zero entries in each row, while the rows are of length $m-n+1$.
Thus, an improvement of our lower bound to $2^{\Omega(m\log(n))}$ could possibly result in finding a \emph{sparse} class of matrices with a high $\chi_1(n)$. 
 
\end{remark}

\section{Complexity questions for cycle patterns}
\label{sec:complexity}

We analyze the algorithmic complexity of realizability questions for cycle patterns. 
This entails realizability by integer weights, parity-realizability and the comparison of cycle patterns. 
These problems turn out to be very difficult in general.

To argue in a meaningful way about computational complexity, we assume that $\psi$ can be evaluated in time polynomial in the size of the input (cycle pattern function and graph).
 To be precise, we assume that $\psi$ is given as a Boolean circuit of size polynomial in the number of edges and vertices of the graph. 
The input of the Boolean circuit is given by the (Boolean) incidence vector of the cycle and the output should be able to encode the four possibilities: the input was not a valid cycle, or $\psi(C)$ is $-, 0 ,+$, respectively. This can be represented by two output gates.

 A well-known fact about Boolean circuits is this: 
 if a Turing machine can decide a language in time $O(t(n))$, then for each $n$ there is a circuit with $g(n)$ gates deciding the language for all inputs of size $n$, and where $g(n)=O(t^2(n))$ (see for example the proof of \cite[Thm 9.30]{sipser1996introduction}). 
 Furthermore, this construction is algorithmic, in the sense that there is an algorithm working in $\mathrm{poly}\bigl(t(n)\bigr)$ complexity that, given a number $n$ in unary, constructs the corresponding circuit.
 Therefore, every time we construct a cycle pattern $\psi$, we will only need to argue that it is computable in polynomial time (on a Turing machine). 
 Finally, we remind the reader of the definitions of some complexity classes which we use in the following sections.
 \begin{itemize}
     \item $RP$: The class of languages $L$ for which there exists a probabilistic Turing machine such that:
     \begin{itemize}
         \item The largest possible runtime is bounded by a polynomial of the input size. 
         \item If the input $x$ is in $L$, the machine accepts $x$ with probability $p\geq \frac{1}{2}$.
         \item If $x\notin L$, the machine rejects $x$ with probability 1.
     \end{itemize}
     \item $\Sigma_{2}^{P} := NP^{NP}$: The class of problems solvable in polynomial time on a nondeterministic Turing machine, which has access to an oracle for some $NP$-complete problem.
 \end{itemize}

\subsection{Realizability}

First, we address one of the most natural problems, which is to check realizability. 

\begin{definition}[Realizability problem] \label{def:realizability-problem}
    Given a digraph $G$ and a cycle pattern $\psi$ represented by a Boolean circuit, is there a weight function $w:E(G)\to\mathbb{Z}$, such that $\psi_w = \psi$?
\end{definition}

We do this using our characterization results from Section \ref{sec:realizability}. 

\begin{theorem}\label{thm:realizablecomplexity}
    Deciding whether a cycle pattern $\psi$ given by a Boolean circuit is realizable for a digraph $G$ is a coNP-complete problem.
\end{theorem}
\begin{proof}  
    First, we show that the problem is in coNP. 
    From Section \ref{sec:witnesssize} we know that if $\psi$ is not realizable, then we can find a non-realizability witness with at most $m-n+2$ distinct cycles. 
    In fact, we can simply give this set of cycles as a certificate: using this certificate, we can verify in polynomial time that the cycle pattern is not realizable. 
    For this, we consider the realization cone from \cref{def:realization+cone}. 
    Let $A'$ and $B'$ be the matrices formed by those rows of $A$ and $B$ corresponding to the cycles comprising the non-realizability witness (i.e. positive and negative versions of the characteristic vectors of the cycles). 
    One can verify in polynomial time that there are $y,z$ (of appropriate length) such that $y^TA'+z^TB'=0, y\geq 0, \bm{1}^Ty\geq 1$ by solving a linear program.
    That shows then that there is a non-realizability witness, so by \cref{thm:char1}, $\psi$ is not realizable. Hence our problem is in coNP.

    To show coNP-hardness, we consider the (coNP-hard) problem of checking if some undirected graph $U$ contains no Hamiltonian cycle. For the reduction, let $G$ be the directed version of $U$ obtained by replacing each undirected edge with two oppositely directed edges. 
    Let $\psi$ be defined by
    \[
    \psi(C)=\begin{cases}
        + & C\text{ is a Hamiltonian cycle}\\
        - & C\text{ is not a Hamiltonian cycle}
    \end{cases} \enspace .
    \]
    This function $\psi$ can be computed by a Boolean circuit of polynomial size as we can easily construct a Turing machine that computes $\psi(C)$ in polynomial time, given as input the vector~$\chi(C)$.
    
    If $U$ is Hamiltonian, say $C$ is a Hamiltonian cycle, then let $C_1,C_2$ be the two directed versions of $C$ in $G$. 
    Let $D_1, D_2, \ldots, D_k$ be all the directed 2-cycles related to each individual edge of $C$ (assuming there are more than 2 vertices in the graph, these are different from $C_1,C_2$). 
    Now these sequences of cycles are an opposing pair, fulfilling the conditions from the third statement in \cref{thm:char1}. 
    Hence, $\psi$ is not realizable.

    Conversely, if $U$ is not Hamiltonian, then $\psi$ is realizable, since it only contains $-$-cycles (for instance, we can give every edge weight $-1$). This completes the proof of coNP-completeness.
\end{proof}

The attentive reader may have noticed that the previous result also holds for parity-realizability: 
a non-Hamiltonian graph gives a graph with only $-$-cycles, which is parity-realizable, and a Hamiltonian graph yields a non-realizable pattern, which is therefore also non-parity-realizable. 
However, we can get a stronger result for parity-realizability. Specifically, we show that the following problem is coNP-hard.

\begin{definition}[Parity-realizability problem]
    Given a digraph $G$ and a weight function $w:E(G)\to\mathbb{Z}$, is the induced cycle pattern $\psi_w$ parity-realizable?
\end{definition}

In contrast to \cref{def:realizability-problem}, we get a realization for the cycle pattern, since the weight also determines a Boolean circuit for $\psi_{w}$. 
From \cref{thm:realizablecomplexity} we know that it was coNP-complete even to decide the existence of this, so we have much more information. 
\begin{theorem}\label{thm:parityrealizabilitycomplexity}
    The parity-realizability problem is coNP-complete.
\end{theorem}
\begin{proof}
To show that the problem is in coNP, let $w$ be a weight function whose cycle pattern is not parity-realizable. 
By \cref{thm:charPG}, either there is a $0$-weight cycle, or there is a nonempty set $S\subseteq E$ such that for each edge $e\in S$, there are cycles $C_e$ and $D_e$ containing $e$ with $\psi(C_e)=+$ and $\psi(D_e)=-$. 
Hence, we have a polynomial-time checkable certificate, with either the cycle of weight $0$, or the subset $S$ and a pair of cycles for each edge $e \in S$. 
Therefore, our problem is in coNP. 

To show that the problem is coNP-hard, we give a reduction from the coNP-hard problem of checking if a given simple digraph contains no Hamiltonian cycle. 
Suppose we have a simple digraph $G$ with $n$ vertices. 
We may assume that the graph is strongly connected (otherwise we could trivially conclude that there is no Hamiltonian cycle). 
Now we construct a graph $G'$ from $G$ as follows: first choose an arbitrary vertex $v\in V(G)$, and add a copy $v'$ of $v$ with the same neighborhood as $v$. 
We also add seven new vertices: $s,s',s'',t,t', t''$ and $u$, see \cref{fig:PGrealizable}. 
We add edges $(v,t''), (t'',t'),(t',t), (t',u), (u,t), (u,s'),(t,s), (s,s'), (s,u), (s',s'')$ and $(s'',v')$. 
For every vertex $x\in V(G)\cup\{v'\}$, we also add edges $(x,t')$ and $(s',x)$.
    \begin{figure}[htb]
        \centering
        \includegraphics[width=0.7\linewidth]{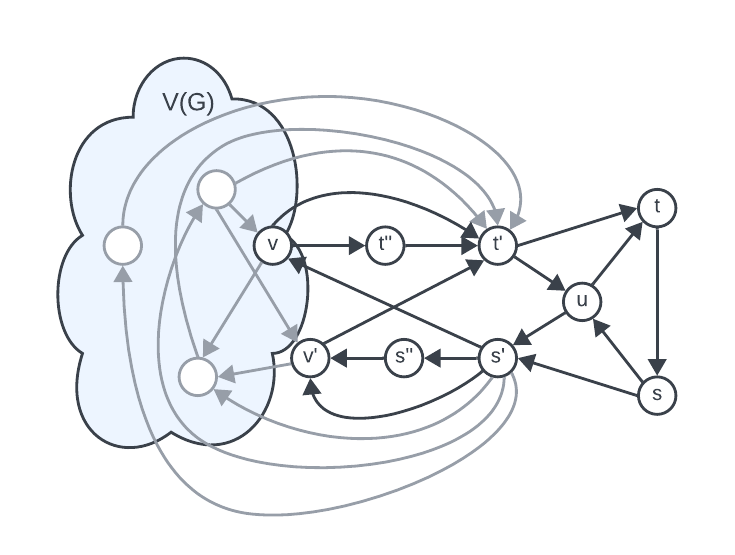}
        \caption{The graph $G'$ constructed from $G$. In black we have vertex $v$ and the eight new vertices, and in grey some examples of vertices of $G$ and their related edges in $G'$.}
        \label{fig:PGrealizable}
    \end{figure}
Now we argue that $G'$ has a Hamiltonian cycle if and only if $G$ has a Hamiltonian cycle. 

\smallskip

Suppose $G$ has a Hamiltonian cycle and let $(v,v'')$ be an edge of this cycle. 
Then, replacing $(v,v'')$ by the path $v,t'',t',u,t,s,s',s'',v',v''$ yields a Hamiltonian cycle for $G'$. 
On the other hand, suppose that $G'$ has a Hamiltonian cycle, and suppose it has edge $(v',v'')$. 
The vertices $t''$ and $s''$ have only one outgoing and incoming edge, and $t$ only has one outgoing edge and two incoming edges. 
This implies that the Hamiltonian cycle must either contain the path $v,t'',t',u,t,s,s',s'',v', v''$ or $v,t'',t',t,s,u,s',s'',v', v''$. 
In both cases we can replace this path by the edge $(v,v'')$\footnote{Note that $v''\neq t'$ since $t''$ has to be the predecessor of $t'$ in a Hamiltonian cycle.}, to obtain a Hamiltonian cycle for $G$.

Now let $w \colon E(G')\to\mathbb{Z}$ be given by $w(t,s)=2n+13$ and by $w(e)=-2$ if $e\neq (t,s)$. 
We argue that the induced cycle pattern $\psi_w$ is parity-realizable if and only if $G'$ contains no Hamiltonian cycle. Suppose that $G'$ has no Hamiltonian cycle. Since $G'$ has $n+8$ vertices, this means that any cycle has at most $n+7$ edges, and therefore any cycle containing $(t,s)$ has weight at least $(2n+13)-2(n+6)=1 > 0$. Any cycle without $(t,s)$ has negative weight. Moreover, if we make a parity game on $G'$ by giving edge $(t,s)$ priority $2$ and every other edge priority $1$, every cycle will have an even highest priority if and only if it contains the edge $(t,s)$. So the cycle pattern $\psi_w$ is parity-realizable if $G'$ has no Hamiltonian cycle.

\smallskip
    
On the other hand, suppose that $G'$ contains a Hamiltonian cycle. We argue that then every edge is contained in both a $+$-cycle and a $-$-cycle. 
To start, consider any edge $(r,r')$ with \break $r,r'\in V(G)\cup \{v'\}$. 
Then, the cycle $r',t',u,s',r$ has negative weight, while $r',t',t,s,s',r$ has positive weight. 
Likewise, if $r$ is a vertex in $V(G)\cup\{v'\}$, then the cycle $r,t',u,s'$ has negative weight, while $r,t',t,s,s'$ has positive weight, so the claim is true for the edges $(r,t')$ and $(s',r)$. 
It remains to check eleven edges involving the nine vertices $v,v',s,s',s'',t,t',t'',u$. 
Among these edges, there are six which do not involve the vertices $t$ and $s$. 
For each of these six edges, it is easy to check that they belong to a $+$-cycle and to a $-$-cycle. 
For instance, the edge $(t',u)$ belongs to the cycle $t',u,s',v,t''$ which has negative weight and to $t',u,t,s,s',v,t''$ which has positive weight. 
The remaining five edges are $(t',t),(u,t),(s,u),(s,s')$, and $(t,s)$. 
Again, it is easy to see that they all belong to a $+$-cycle; for instance, the edge $(u,t)$ belongs to the cycle $u,t,s$ and $(t',t)$ belongs to $t',t,s,s',v,t''$. 
Therefore, the only remaining part is to show that these five edges belong to some $-$-cycle. 
To see that this is the case, note that any Hamiltonian cycle of $G'$ has weight $-1$. 
Moreover, it contains either the path $t',t,s,u,s'$ or the path $t',u,t,s,s'$. 
Even more, if a Hamiltonian cycle in $G'$ contains one of these two paths, then there exists a Hamiltonian cycle in $G'$ that contains the other path (we can just replace one path by the other in the cycle). 
These two Hamiltonian cycles have weight $-1$ and contain the five edges. 
Hence, every edge is contained in both a $+$-cycle and a $-$-cycle, which implies that $\psi_w$ is not parity-realizable by \cref{thm:charPG} (taking $S=E$). 
Since the construction of $G'$ can be done in polynomial time, we conclude that we have a valid reduction, hence the parity-realizability problem is coNP-complete.
\end{proof}

\subsection{Distinguishing cycle patterns}

Reducing the size of a realization is interesting from an algorithmic perspective: suppose we have some procedure to find smaller weights in the graph while maintaining the crucial structure. 
Combining this with a pseudopolynomial algorithm for some problem in the graph provides a method to solve it more efficiently. 
This motivates the following.  

\begin{definition}[Bounded realization problem]
    Given a weighted digraph $(G,w)$ and an integer $k$, decide if there exist integer weights $w'$ bounded by $k$ in absolute value, such that the graphs $(G,w)$ and $(G,w')$ have the same induced cycle pattern.
\end{definition}

We have already seen in \cref{thm:expweights} that there are cycle patterns for which $k$ has to be exponentially big. 
We extend this idea further in \cref{sec:geometric}. 
Here, we consider the much simpler problem of checking if such a weight reduction even gives a correct answer. 
First, we introduce the related basic problem of finding cycles of weight zero as a building block. 

\begin{definition}[Zero-weight cycle problem]
    Given a digraph $G=(V,E)$ and a weight function $w:E\to \mathbb{Z}$, is there a cycle $C\subseteq E$ such that $w(C):=\sum_{e\in C}w_e=0$?    
\end{definition}

One can show NP-completeness of this similar to \cite{kawase2015finding}, where in their introduction they show NP-completeness of finding paths of weight zero. 
We include a proof for completeness. 

\begin{lemma}\label{lem:zeroweightcycle}
    The zero-weight cycle problem is NP-complete.
\end{lemma}
\begin{proof} 
    We reduce from the Hamiltonian path problem with fixed start vertex $s$ and end vertex $t$ (this NP-complete problem asks whether there exists an $s,t$ path visiting all vertices in a digraph $G$). 
    If we have a digraph $G$ with specified vertices $s$ and $t$, we can add edge $(t,s)$ with weight $|V(G)|-1$ and give all other edges weight $-1$. 
    Clearly a Hamiltonian $s,t$ path in $G$ leads to a zero weight cycle in $G'$ and vice versa. 
    So we reduced Hamiltonian path to the zero-weight cycle problem. 
    Moreover, the latter problem is clearly in NP as we could just give the cycle as certificate.
\end{proof}

To check if a solution to the bounded realization problem is correct, one can compare the sign of each cycle with the induced sign. 

\begin{definition}[Induced cycle pattern distinction problem]
    Suppose we are given a digraph $G=(V,E)$ and two weight functions $w_1:E\to \mathbb{Z}$ and $w_2:E\to \mathbb{Z}$. 
    Is there a cycle $C\subseteq E$ such that $\psi_{w_1}(C)\neq \psi_{w_2}(C)$? 
\end{definition}

We show that the induced cycle pattern distinction problem is NP-complete using the zero-weight cycle problem.

\begin{theorem}
    The induced cycle pattern distinction problem is NP-complete.
\end{theorem}
\begin{proof}
Suppose we have a digraph $G=(V,E)$ with weight function $w:E\to\mathbb{Z}$. 
Fix an edge $d$, and consider the weight functions $w_1':E\to\mathbb{Z}$ and $w_2^d:E\to\mathbb{Z}$ defined by 
\begin{equation*}
    w_1'(e)=(|V+1|)w(e)+1, \quad w_2^d(e)=\begin{cases}
        (|V|+1)w(e)+1 & e\neq d\\
        (|V|+1)(w(e)-1)+1 & e=d
    \end{cases} \qquad \text{ for } e\in E \enspace .
\end{equation*}

Note that for both $w_1'$ and $w_2^d$, there are no cycles of weight exactly $0$. 
This follows since all edge weights are congruent 1 modulo $|V|+1$, and since the cycle lengths are at most $|V|$ this means that its weight cannot be divisible by $|V|+1$. 
Furthermore, we have $w_1'(C)>0\Leftrightarrow w(C)\geq 0$ and $w_1'(C)<0\Leftrightarrow w(C)<0$. 
Moreover, we have $w_2^d(C)> 0 \Leftrightarrow (w(C)>0)\vee(w(C)=0\wedge d\notin C)$ and $w_2^d(C)< 0 \Leftrightarrow (w(C)<0)\vee(w(C)=0\wedge d\in C)$. 
If there is a cycle with a different sign for $w_1'$ and $w_2^d$, this means that there is a cycle $C$ with $w_1^d(C)>0$ and $w_2^d(C)<0$, and this can only happen if $w(C)=0$ and $d\in C$. 
If the induced cycle pattern distinction problem on $w_1'$ and $w_2^d$ is a ``No''-instance, we know that there is no cycle $C$ containing $d$ with $w(C)=0$.
    
Now we describe the reduction. 
Let $G'$ be the disjoint union of $|E|$ copies of $G$. 
Let $w_1$ be a weight function for $G'$ with the same value as $w_1'$ for each of the copies. 
Let $w_2$ be the weight function equal to $w_2^d$ on the $d$-th copy of $G$. 
If we solve the induced cycle pattern distinction problem on $G'$, this is the same as solving it on $G$ with $w_1'$ and $w_2^d$ for all $d$. 
Hence, if we have a ``Yes''-instance for the problem on $G'$, this implies that there was a zero cycle in $G$ for $w$.
On the other hand, if it is a ``No''-instance, this implies that for each edge, there is no zero-cycle of $G$ containing it.
Therefore, there was no zero-cycle at all. 
That completes our reduction. 
Note that, while the graph $G'$ as described is not strongly connected, one can modify the construction slightly without changing the argument, namely we may add edges with large weights between the disconnected components of $G'$. 

Finally, since there is a trivial certificate for ``Yes''-instances of the cycle pattern distinction problem (just calculate the weights of the cycle), we conclude that the induced cycle pattern distinction problem is NP-complete. 
\end{proof}

This immediately leads to the following result:
\begin{lemma}
    The bounded realization problem is in $\Sigma_2^{P}$.
\end{lemma}
\begin{proof}
If the answer is ``Yes'', then a nondeterministic machine can just guess the correct weights and ask the oracle if the two graphs have the same cycle pattern.   
\end{proof}

In the bounded realization problem, we try to optimize a problem for which checking validity of the answer is already $coNP$-complete. There are numerous problems of this kind that turn out to be $\Sigma_2^{P}$-complete, see for example the optimization problems in \cite{schaefer2002completeness}. This motivates the following conjecture.

\begin{conjecture}
    The bounded realization problem is $\Sigma_2^{P}$-complete.
\end{conjecture}

\section{Games on graphs}
\label{sec:games}

We turn to the application of cycle patterns and the insights gathered in the previous sections. 
We start by introducing three classes of games: mean payoff games, energy games, and parity games.

\subsection{Basics}

For a more in-depth treatment of games on graphs we refer to \cite{fijalkow2023gamesgraphs}. 
Each of these games is played on an \emph{arena} 
$G=(\Vmax \sqcup \Vmin,E)$, 
a directed graph with vertices divided into two disjoint subsets: 
the set $\Vmax$ of vertices controlled by player Max and the set $\Vmin$ of vertices controlled by player Min. 
We assume that every vertex of the graph $G$ has at least one outgoing edge. 
A \emph{weighted arena} is a pair of an arena $G$ and weight function $w : E \to \RR$ on the edges of the underlying graph; for algorithmic purposes, we usually assume that $w$ only attains integral values. As in the previous sections, we may also interpret the weight function as a vector in $\RR^E$ instead.
Fixing a start point, this defines a (mean payoff) game.   
The game is played as follows: 
a pebble is placed on some initial vertex $v_0$, and the player who controls $v_0$ moves the pebble along one outgoing edge $(v_0,v_1)$. 
Next, the controller of $v_1$ picks an outgoing edge $(v_1,v_2)$, and so on, continuing indefinitely. 
This procedure yields an infinite sequence of edges $\pi = (e_1,e_2,e_3,\ldots)$. 
The \emph{outcome} of the game is based on this sequence, but is defined differently for different games. 
For mean payoff games (MPGs) the outcome is the average weight of the edges encountered:
\[
O(\pi)=\limsup_{t\to\infty}\frac{1}{t}\sum_{i=1}^tw(e_i) \, .
\]
Player Max wants to maximize the outcome, while player Min wants to minimize it. 
For energy games (EGs), the outcome is the smallest cumulative sum of weights
\[
O(\pi)=\inf_{t \ge 1} \sum_{i=1}^tw(e_i) \, .
\]
In this case, the outcome may be equal to $-\infty$. 
For parity games (PGs), the outcome is
\[
O(\pi)=\begin{cases}
    1 & \text{if }\limsup_{i\to\infty}(w(e_i))=0 \mod{2},\\
    -1 & \text{if }\limsup_{i\to\infty}(w(e_i))=1 \mod{2} \, .
\end{cases}
\]
In other words, player Max wins a parity game if the largest edge weight that occurs infinitely often is even, and player Min wins if the largest edge weight that occurs infinitely often is odd. 
In the context of parity games, we refer to the weights as \emph{priorities}. 
In the literature on parity games, the priorities are usually on the nodes, but this does not make a fundamental difference: 
a parity game with priorities on the nodes can easily be transformed into an equivalent one with priorities on edges, and vice versa.

A \emph{(positional) strategy} of
player Max is a function $\sigma \colon \Vmax \to V$ that associates a vertex reachable in a single move to every vertex controlled by Max, i.e., it is a function that satisfies $(i,\sigma(i)) \in E$ for all $i \in \Vmax$. 
We define a strategy $\tau \colon \Vmin \to V$ of player Min analogously. 
We say that player Max plays according to $\sigma$ if they move the pebble to $\sigma(i)$ whenever it lands on $i$. 
If we fix an initial vertex $v_0$ and suppose that players play according to strategies $(\sigma,\tau)$, then the entire movement of the pebble $\pi = \pi(\sigma,\tau,v_0) = (e_1,e_2,e_3,\ldots)$ is determined. 
More precisely, the pebble goes to some cycle of the graph $G$ and stays there forever. 
It is known that all three classes of games are \emph{positionally determined}. 
That means that there exists a function $\val \colon V\to \QQ \cup\{-\infty\}$ and a pair of strategies $(\sigma^*,\tau^*)$ that for any vertex $v_0$ and every pair $(\sigma,\tau)$ satisfy
\begin{equation}\label{eq:optimal_policies}
O\bigl(\pi(\sigma,\tau^*,v_0)\bigr) \le \val(v_0) \le O\bigl(\pi(\sigma^*,\tau,v_0)\bigr) \, .
\end{equation}
In other words, by playing according to $\sigma^*$ player Max guarantees that the outcome of the game is at least $\val(v_0)$ and by playing according to $\tau^*$ player Min guarantees that the outcome of the game is at most $\val(v_0)$. 
Such strategies $\sigma^*,\tau^*$ are called \emph{optimal} and the function $\val$ is called the \emph{value} of the game. We note that the inequality from \cref{eq:optimal_policies} is true even for non-positional strategies $(\sigma,\tau)$, but, for the sake of simplicity, we only consider positional strategies in this work.

We define the problem of solving an MPG as the problem of finding the set of vertices with nonnegative value. Following \cite{DBLP:journals/dam/BjorklundV07}, we call the resulting partition of vertices the \emph{zero-mean partition} of the arena.
Likewise, solving an EG means finding the set of vertices with finite value, and solving a PG means finding the set of vertices that have value equal to $1$.\footnote{More generally, one can consider the problem of finding the value and a pair of optimal strategies in each of these games. 
However, these problems are polynomial-time (Turing) equivalent to the problems that we consider in this work, see, e.g., \cite{DBLP:journals/dam/BjorklundV07} for the details in the case of MPGs.} 
The following definition and a well-known lemma relate the problems of solving games to cycle patterns. 
We use the following notation: 
if $\sigma$ is a strategy of player Max, we denote $E_\sigma = \{(i,j) \in E, i \in \Vmin\} \cup \{\bigl(i,\sigma(i)\bigr) \colon i \in \Vmax\}$.

\begin{definition}[Zero-mean partition problem]
    Given an arena $G=(\Vmax\cup \Vmin, E)$ and a (parity-)realizable cycle pattern $\psi$, determine all vertices $v$ with the following property:
    \begin{itemize}
        \item[] There exists a strategy $\sigma$ of player Max, such that every cycle $C$ reachable from $v$ in the graph $(V,E_{\sigma})$ has $\psi(C)\in \{0,+\}$.
    \end{itemize}
We denote the set of all such vertices by $\Win(\psi)$, called the \emph{winning region}. 
\end{definition}

\begin{lemma}\label{lem:solvinggames}
If $w$ is a weight vector that realizes $\psi$, then the set $\Win(\psi)$ coincides with the set of vertices that have nonnegative value in the MPG played on $(V,E,w)$ and with the set of vertices that have finite value in the EG played on $(V,E,w)$. 
Moreover, if $w$ is a vector of priorities that parity-realizes $\psi$, then $\Win(\psi)$ coincides with the set of vertices that have value $1$ in the PG played on $(V,E,w)$.
\end{lemma}
\begin{proof}
To prove the first claim, consider the MPG played on $(V,E,w)$. 
Let $v \in \Win(\psi)$ and let $\sigma$ be such that every cycle $C$ reachable from $v$ in the graph $(V,E_{\sigma})$ has $\psi(C)\in \{0,+\}$. 
Then, $\val(v) \ge O\bigl(\pi(\sigma,\tau^*,v)\bigr) \ge 0$, because the cycle reached by the sequence $\pi(\sigma,\tau^*,v)$ satisfies $\psi(C)\in \{0,+\}$, hence this cycle has nonnegative weight. 
Conversely, if $\val(v) \ge 0$, then the strategy $\sigma^*$ satisfies $O\bigl(\pi(\sigma^*,\tau,v)\bigr) \ge 0$ for all $\tau$, so any cycle reachable from $v$ in $(V,E_{\sigma^*})$ must satisfy $\psi(C)\in \{0,+\}$. 
The remaining claims follow analogously.
\end{proof}

Recall from Section \ref{sec:terminology} that one can easily construct a realization for a parity-realizable cycle pattern from a parity-realization. Therefore, \cref{lem:solvinggames} immediately gives us a (well-known) reduction from solving a PG to solving an MPG or an EG.

Finally, in general, games can be played on graphs which are not strongly connected, while our analysis of cycle patterns is restricted to strongly connected graphs. 
However, restricting attention to only strongly connected arenas does not make a significant difference for solving games. 
Indeed, if an arena is not strongly connected, then one can find the zero-mean partition by repeatedly finding the strongly connected components of the arena, solving the zero-mean partition on one of the final components, propagating it using attractor computations, and simplifying the graph. 
The details of this procedure are discussed in \cite[Section~3.1]{Friedmann2009ParityPractice} in the context of parity games and are valid for the problem of finding the zero-mean partition. 
Hence, from now on we assume that the games are played on strongly connected arenas.

\subsection{Solving games using cycle patterns} \label{sec:solving-games-via-cycle-patterns}

As we have seen in \cref{lem:solvinggames}, the cycle pattern already suffices to determine the zero-mean partition. 
However, it has strictly less information than the weight function. 
We examine the hardness of computing the zero-mean partition from access to the cycle pattern. 

First, we consider the most general setting, where we have an arena $G=(\Vmax\cup \Vmin)$ and a cycle pattern $\psi$, and we ignore realizability for now. 
Of course, we want to find the nodes with nonnegative value. 
However, without the realizability assumption, the game might not have a well-defined value. 
Therefore, we use an asymmetric definition of the zero-mean partition problem:

\begin{definition}[General zero-mean partition problem]
     Let $\psi$ be a (not necessarily realizable) cycle pattern of the digraph $G=(\Vmin\cup \Vmax,E)$ given by a Boolean circuit. 
     Let $v_0$ be an initial node, is there a positional Maximizer strategy $\sigma$ such that, if the Maximizer plays according to $\sigma$, every cycle reachable from $v_0$ is nonnegative?
\end{definition}

To show our hardness result, we need the following well-known $\Sigma_{2}^P$-complete problem (see \cite[Cor. 6]{WRATHALL197623}).

\begin{definition}[$\exists\forall$ SAT]
    Given a Boolean formula 
    $\phi(x,y)$ (where $x$ and $y$ are Boolean vectors), is it true that $(\exists x)(\forall y) \phi(x,y)$? 
\end{definition}

 This allows us to deduce the following theorem.
 
\begin{theorem} \label{thm:general-zero-mean-partition-problem}
    The general zero-mean partition problem is $\Sigma_{2}^{P}$-complete.
\end{theorem}
\begin{proof}
    First of all, the general zero-mean partition problem is in $\Sigma_{2}^{P}$. If we have a ``Yes''-instance, then we can give as a certificate a winning positional strategy $\sigma$ for the maximizer. We can then construct the graph $G_{\sigma}=(V,E_{\sigma})$ (recall that $E_\sigma = \{(i,j) \in E, i \in \Vmin\} \cup \{\bigl(i,\sigma(i)\bigr) \colon i \in \Vmax\}$). Then we check that this is indeed a winning strategy by checking that there are no negative cycles for $\psi$ in $G_{\sigma}$; the latter problem is obviously in coNP.
    
    Next, we show that the general zero-mean partition problem is $\Sigma_{2}^P$-hard. We do so by reducing from $\exists\forall$ SAT in polynomial time. Suppose we have a Boolean formula $\phi(x,y)$, with $x\in \{0,1\}^k$ and $y\in \{0,1\}^l$. We create a digraph $G$ as follows: we have $\Vmax=\{v_1,v_2, \ldots, v_k\}$ and $\Vmin=\{v_1',v_2', \ldots, v_l'\}$. For $i=1,2,\ldots, k-1$, we add two edges from $v_i$ to $v_{i+1}$, $e_i$ and $\hat{e}_i$, which correspond to $x_i$ being true and false, respectively. We add edges $e_k$ and $\hat{e}_k$ from $v_k$ to $v'_1$. Likewise, we add edges $e_j'$ and $\hat{e}_j'$ from $v_j'$ to either $v_{j+1}'$ if $j<l$ or to $v_1$ if $j=l$. Similarly, we associate $v'_i$ being true with edge $e_i'$ and $v'_i$ being false with edge $\hat{e}_i'$. See \cref{fig:sigma2game} for an illustration.

    \begin{figure}[H]
        \centering
        \includegraphics[width=\linewidth]{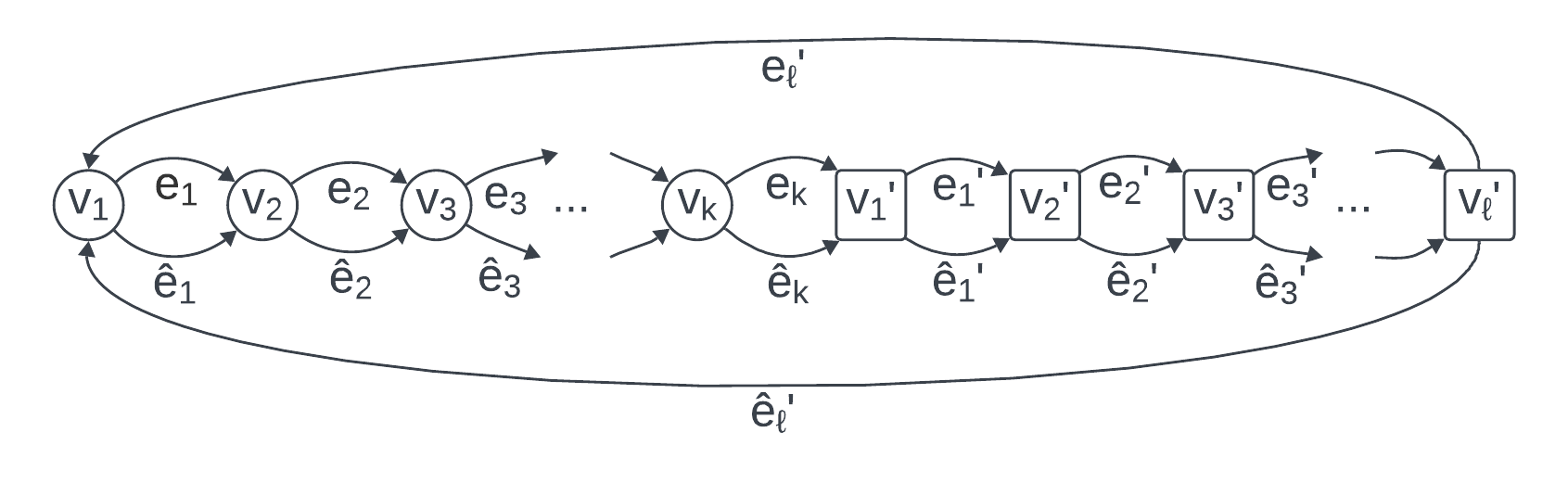}
        \caption{A visualization of the arena of the game for the proof of \cref{thm:general-zero-mean-partition-problem}. $\Vmax$ is depicted by circles, and $\Vmin$ by squares.}
        \label{fig:sigma2game}
    \end{figure}
    Note that there is a one-to-one correspondence between the cycles in this graph and the truth assignments for $(x,y)$. We now define $\psi$ as
    \[
    \psi(C)=\begin{cases}
        + & C\text{ corresponds to }(x,y)\text{ and }\phi(x,y)=1\\
        - & C\text{ corresponds to }(x,y)\text{ and }\phi(x,y)=0 \, .
    \end{cases}
    \]
    The circuit of $\psi$ can easily be constructed in polynomial time. 
    If the original $\exists\forall$ SAT problem is a ``yes''-instance, then there is an $x$ that makes $\phi$ true for all $y$. 
    This implies that the related maximizer strategy does not allow for any $-$-cycles, hence it yields a nonnegative value for all starting nodes in the graph. 
    In the other direction, if there is a positional strategy that gives nonnegative payoff for some starting node $v_0$, then this gives a valid $x$ for the $\exists\forall$ SAT problem. 
    So we have indeed given a reduction, and that completes the proof.
\end{proof}

It is maybe not too surprising that the general problem is difficult, considering the lack of structure. 
However, it turns out that restricting to realizable patterns does not make finding the winner easy. 
It turns out that even the one-player case becomes difficult to solve if we only consider the cycle pattern. 
To make the one-player case formal, we consider the following problem.

\begin{definition}[One-player zero-mean partition]
    Suppose we are given a strongly connected arena $G=(\Vmax\cup \Vmin,E)$ with $\Vmin=\emptyset$, a starting node $v_0$, and a cycle pattern $\psi$ given by a Boolean circuit, with the promise that $\psi$ is realizable. 
    Does there exist a strategy for which the Maximizer can guarantee reaching a $0$-cycle or a $+$-cycle?
\end{definition}

\begin{theorem}\label{thm:1pZMP}
    There is no randomized polynomial-time algorithm to solve the one-player zero-mean partition problem, unless $NP=RP$.
\end{theorem}
Here, by a ``randomized polynomial-time algorithm'' we mean an algorithm that stops in polynomial time and has the following two properties:
\begin{itemize}
\item it outputs ``No'' with probability $1$ if $\psi$ is realizable and player Max does not have a strategy that guarantees reaching a nonnegative cycle;
\item it outputs ``Yes'' with probability at least $1/\mathrm{poly}(n)$ if $\psi$ is realizable and player Max has a strategy that guarantees reaching a nonnegative cycle.
\end{itemize}
Note that we do not demand anything when $\psi$ is not realizable, other than the fact that the algorithm stops on such instances in polynomial time.
This type of randomized algorithm will allow us to make the connection between (U)SAT and RP in the following proof. 

\begin{proof}
    We show this by performing a polynomial time reduction from USAT (unambiguous SAT). The USAT problem asks, given a SAT formula with at most one satisfying assignment, whether the SAT formula has a satisfying assignment. 
    It is known that, if there is a randomized polynomial time algorithm for USAT, then $RP=NP$ \cite{ValiantVazirani:1986}.  
    Suppose we have a SAT formula $\phi(x)$ with $n$ variables, and with at most one satisfying assignment. 
    Let $G$ be the complete directed graph on $n$ vertices. 
    Assume that the vertices of $G$ are labeled $1,2,\ldots, n$. 
    For all cycles $C$ of $G$, let $S(C)=(v_1, v_2, \ldots, v_k)$ be an ordering of the vertices of $C$ in such a way that $v_1$ is the vertex 
    with the smallest label of $C$, and for $i\geq 2$, $v_i$ comes after $v_{i-1}$ in the cycle, that is, $(v_{i-1},v_i)$ is an edge in the cycle. 
    While there can be many cycles in $G$ on the same set of nodes, there is always only one cycle $C$ on these nodes where the labels in $S(C)$ form an increasing sequence.

    Let $x_C$ be the Boolean vector of length $n$ with $x_i=1$ if and only if the vertex $i$ is in $C$. 
    Now we define a cycle pattern $\psi$ for $G$:
    \[
    \psi(C)=\begin{cases}
        + & \text{the sequence $S(C)$ is increasing and } \phi(x_C)=1,\\
        - & \text{otherwise}.
    \end{cases}
    \]
    The circuit of $\psi$ can be constructed in polynomial time from $\phi$. 
    Moreover, since there is at most one satisfying assignment for $\phi$, there is at most one cycle $C$ with $\psi(C)=+$. 
    This implies that $\psi$ is realizable: if there are no $+$-cycles, the pattern can be realized by giving each edge as weight $-1$; otherwise, we can give all the edges on the one $+$-cycle weight $1$, and all other edges weight $-n$, which will also be a realization.
    Since the graph is strongly connected, the zero-mean partition only depends on the existence of a $+$-cycle: 
    If such a cycle exists, then player Max can just use the edges of the cycle in their strategy, and make sure that the pebble ends up in the cycle. Now solving the zero-mean partition lets us solve USAT:
    if there is a $+$-cycle, this implies that $\phi$ must have had a satisfying assignment. Likewise, if there are no $+$-cycles, there was no satisfying assignment. 
    So indeed we found a polynomial time reduction from USAT to the one-player zero-mean partition problem. 
    Hence, if we had a randomized polynomial algorithm for the one-player zero-mean partition problem, the following algorithm would be a randomized polynomial algorithm for USAT:
    \begin{itemize}
        \item Construct the cycle pattern $\psi$ for graph $K_n$.
        \item Run the randomized polynomial algorithm for one-player zero-mean partition.
    \end{itemize}
    Since both steps are in polynomial time, and this new algorithm fulfills the requirements on the output of a randomized algorithm stated in \cite{ValiantVazirani:1986}, we conclude that this would be a randomized polynomial algorithm for USAT. 
    Hence this implies $NP=RP$. 
    This completes the proof.
\end{proof}

We leave the question about the complexity of finding the zero-mean partition of realizable patterns for two-players for further work.

\section{Geometric hardness for mean payoff games}
\label{sec:geometric}

As mentioned before, numerous different algorithms for solving mean payoff games and energy games have been proposed. 
For many of these algorithms, there is a pseudopolynomial upper bound on the running time: 
it can be bounded by a polynomial in $|V|$, $|E|$, and $W$, where the latter is an upper bound on the absolute value of the edge weights in the arena.

Having a pseudopolynomial algorithm for mean payoff games, it may be tempting to aim for a small representation of a game in the spirit of the seminal result by Frank and Tard\'os~\cite{FrankTardos:1987}. 
Their reduction takes as input a rational vector $w = (w_1, \ldots, w_m)$ and an integer $N$. 
It returns an integral vector $\overline{w} = (\overline{w}_1,\ldots,\overline{w}_m)$ such that $||\overline{w}||_{\infty} \leq 2^{4m^3}N^{m(m+2)}$ and $\sgn(w \cdot b) = \sgn(\overline{w} \cdot b)$, for all integer vectors $b$ with $||b||_1\leq N-1$. 
Note that this reduction can even be computed in polynomial time. 
One could try to apply this reduction to a mean payoff game where $w$ is the vector of weights (independent of the graph structure). 
Then the incidence vector of the cycles could be considered as the vectors $b$ giving rise to a reduction preserving the cycle pattern. 
As such, it would also preserve the winning region. 
However, the resulting upper bound on the reduced weights for a game with $m$ edges would amount to $2^{4m^3}(k+1)^{m(m+2)}$ where $k$ is the length of a longest cycle in the graph. 
Unfortunately, this is still exponential in the size of the graph. 

As this reduction is a very general tool with many remarkable applications in the realm of linear programming, see e.g.~\cite{Tardos:1986}, one could wonder if it is possible to come up with a reduction more tailored to mean payoff games. 
However, we demonstrate the limitations of such an idea. 
At first, we already saw in \cref{thm:expweights} that any reduction preserving cycle patterns would still result in exponential weights. 
We will extend this in \cref{thm:core-preserving-lower-bound} to reductions preserving the crucial structure for linear decision trees. 
Now, we start by excluding a reduction that does not care about the graph structure but only about the zero-mean partition. 

\subsection{Reductions preserving the zero-mean partition}

Assume we want a reduction that only preserves the zero-mean partition of the game. If we get both the weights and the graph structure as input, then solving the game allows to determine the winning regions and one can just use weights from $\{-1,0,1\}$, as we will see in \cref{sec:cycle-patterns-arrangements}. However, as our reduction is meant to be a preprocessing, in this section we restrict attention to reductions that only use the weights but not the graph structure. The reduction of Frank and Tard\'os has this property and results in weights of exponential size. We show that this cannot be avoided. To do so, we start with the following lemma.

\begin{lemma}\label{lem:exponential_integers}
Let $k \ge 1$ and suppose that integer vectors $a,b,c \in \ZZ^k$ satisfy the inequalities
\begin{equation}\label{eq:exponential}
\begin{aligned}
a_1 + \dots + a_{i-1} + b_i + c_i &\geq 0 \, ,\\
b_1 + \dots + b_{i-1} + a_i + c_i &< 0 \, ,\\
\end{aligned}
\end{equation}
for all $i = 1,2,\dots,k$. Then, $\max\{|a_{k}|,|b_{k}|\} \ge 2^{k-2}$.
\end{lemma}
\begin{proof}
We will show by induction that $b_i - a_i \ge 2^{i-1}$ for all $i$. Since $a,b,c$ are integer vectors, for $i = 1$ we have $b_1 \ge - c_1$ and $a_1 \le -1 - c_1$, which gives $b_1 \ge 1 + a_1$. Likewise, for $i > 2$ we get
\[
b_i \ge - c_i - (a_1 + \dots + a_{i-1}) \ge 1 + a_i + (b_1 + \dots + b_{i-1}) - (a_1 + \dots + a_{i-1}) \, .
\]
Hence, by induction, $b_i - a_i \ge 1 + (b_1 - a_1) + \dots + (b_{i-1} - a_{i-1}) \ge 1 + (1+2 + 4 + \dots + 2^{i-2}) = 2^{i-1}$. In particular, we have $\max\{|a_{k}|,|b_{k}|\} \ge \frac{1}{2}|b_k - a_k| \ge 2^{k-2}$.
\end{proof}
\begin{remark}
The vectors defined by $a_i = \lfloor -2^{i-2}\rfloor$, $b_i = \lfloor 2^{i-2}\rfloor$, $c_i = 0$ satisfy \cref{eq:exponential}, so \cref{lem:exponential_integers} is tight (for $k\neq 1$).
\end{remark}

We are now equipped to derive the following impossibility result. 

\begin{theorem}\label{th:multiple_reduction}
    For each even $m\in\mathbb{N}$, there exist weights $w_1, w_2,\ldots, w_m\in \mathbb{Z}$ such that there exist $m$ arenas with $m$ edges each (and with $w(e_i)=w_i$ for each edge $e_i$ in each of the $m$ MPGs), with the following property: if we replace $w_1,w_2, \ldots, w_m$ with integer weights $w_1', w_2', \ldots, w_m'$ such that $\max|w_i'|< 2^{m/2-2}$, this would change the zero-mean partition in at least one of the $m$ MPGs.    
\end{theorem}
\begin{proof}
    Take $w_i=(-2)^i$ for $i=1,2,\ldots, m$. We construct $m$ arenas, by doing the following for $i=1,2,\ldots, \frac{m}{2}$: 
    \begin{itemize}
        \item The $(2i-1)$-th arena has a cycle $C$ controlled by player Max. The cycle $C$ has $i$ edges with weights $w_{2i-1}$ and $w_2, w_4, \ldots, w_{2i-2}$, respectively. Moreover, there is a second cycle $C'$ that contains one vertex $v_1$ of $C$ and for the rest only has vertices controlled by player Min (see \cref{fig:MPG_minweight_multiple}). The edges of $C'$ have weights $w_1, w_3, \dots, w_{2i-3},w_{2i+1},w_{2i+3},\ldots,w_{m-1}$. In particular, all edges of $C'$ have negative weight. Finally, for every index $j$ that we have not used yet, we add some self-loop or chord in $C'$ with weight $w_j$ (their precise location is not important to the proof, we only do not allow such an edge to start at $v_1$). Every vertex in the resulting MPG has negative value: player Min can play the strategy $\tau$ that uses only the edges of $C'$. With this strategy, there are only two cycles possible in the remaining graph $G_{\tau}$, namely $C$ and $C'$, both with negative weight.
        
        Now suppose that $w'(C)\geq 0$. This would imply that player Max can use the edges of $C$ to guarantee a nonnegative outcome starting from $v_1$, but then the zero-mean partition of $w'$ would be different from that of $w$. Hence $w'(C)=w'_{2i-1}+w'_2+w'_4+\ldots + w'_{2i-2}<0$ for $i=1,2,\ldots,\frac{m}{2}$.
        \item The $2i$-th arena has a cycle $C$ controlled by player Min with $i$ edges, which have edge weights $w_{2i}$ and $w_1, w_3\ldots, w_{2i-3}$. There is a cycle $C'$ consisting of $v_1\in C$ and a number of vertices controlled by player Max. The edges of $C'$ have positive weights, namely $w_2,w_4,\dots,w_{2i-2},w_{2i+2},\dots,w_{m}$, and there is a chord or loop for every remaining $w_j$ (again not from $v_1$). In a similar way as before, player Max can use the edges of $C'$ to ensure nonnegative outcome for $w$ from every starting position, since both $C$ and $C'$ have nonnegative weight. This implies that $w'(C)=w'_{2i}+w'_1+w'_3+\ldots + w'_{2i-3}\geq 0$.
    \end{itemize} 
    \begin{figure}[h]
        \centering
        \includegraphics[width=\linewidth]{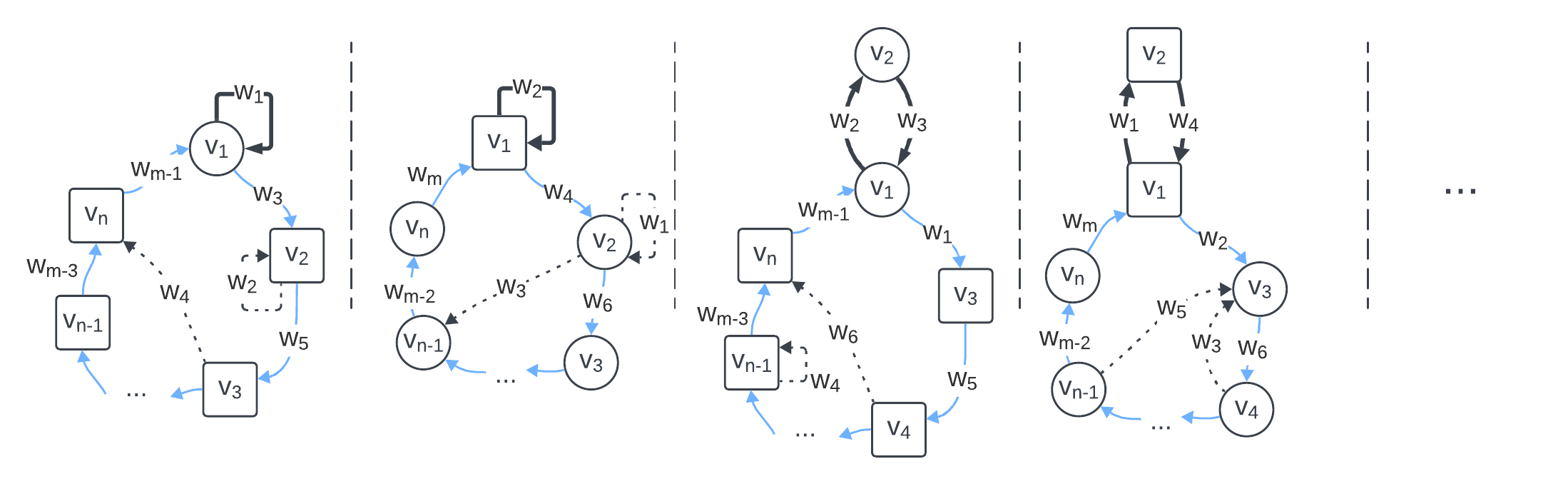}
        \caption{The start of the sequence of graphs constructed in the proof of \cref{th:multiple_reduction}. Circles are vertices controlled by player Max, squares are vertices controlled by player Min. In every graph, the cycle $C$ is marked with bold edges, and the cycle $C'$ by colored edges. There are $n=\frac{m}{2}+i-2$ vertices in each graph.}
        \label{fig:MPG_minweight_multiple}
    \end{figure}
    
    As observed, if $w'$ preserves the zero-mean partition then for every $i\leq \frac{m}{2}$ we have \break $w'_{2i}+w'_1+w'_3+\ldots + w'_{2i-3}\geq 0$ and $w'_{2i-1}+w'_2+w'_4+\ldots + w'_{2i-2}<0$. It then follows from Lemma \ref{lem:exponential_integers} that $\max (|w'_m|,|w'_{m-1}|)\geq 2^{m/2-2}$, which is what we needed to prove.
\end{proof}

\subsection{LDT algorithms}

We now consider algorithms for solving mean payoff games and energy games in the framework of the \emph{linear decision tree model}.
We give a short introduction, see~\cite{BuergisserClausenShokrollahi:1997} as a further general reference, and~\cite{KaneLovettMoran:2019} for recent advances; we will take a geometric point of view similar to~\cite{BjornerLovaszYao:1992}. 
A \emph{linear decision tree} $\dectree$ (LDT) with $m$ inputs and $n$ bits of output is a binary tree where each inner node is associated with a linear inequality $a_1 x_1 + \dots + a_m x_m \geq a_0$ for $a_0, a_1, \dots, a_m \in \ZZ$ and each leaf is labeled by an element in $\{0,1\}^n$. 
For an input vector $(x_1,\dots,x_m) \in \RR^m$, one starts at the root of $\dectree$ and successively descends to the left or right child of a node depending on the validity of the linear inequality associated to the respective node. 
Finally, the output is the element associated with the reached leaf.

Recall that a \emph{polyhedral fan} is a polyhedral complex where all polyhedra are polyhedral cones. 
A \emph{polyhedral subdivision} of a space $S$ is a polyhedral complex such that the union of all polyhedra is $S$. 
We refer to~\cite{Ziegler:1995} for more details on polyhedral geometry. 
However, note that it is sometimes convenient for us to consider polyhedral subdivisions composed of half-open polyhedra so that one obtains a disjoint union. 

An LDT $\dectree$ gives rise to several polyhedral subdivisions of the input space. 
We will focus on LDTs with $a_0 = 0$ for all nodes; then the polyhedral subdivisions are indeed polyhedral fans. 
Since the cycle-inequalities defining the zero-mean partition problem are homogeneous,
we can make this assumption on the LDTs in our context without losing expressive power (see e.g. \cite[Cor. 5.4]{snir1981proving}). 

The finest subdivision is the polyhedral fan arising from the collection of linear forms associated with the nodes of $\dectree$. 
That is, we consider the polyhedral subdivision of $\RR^m$ into polyhedral cones arising as the cells of the hyperplane arrangement given by all linear hyperplanes \break $\setOf{x \in \RR^m}{a^{\top}_v x = 0}$ for inner nodes $v$ of the tree. 
We call this the \emph{node subdivision} and denote the polyhedral fan by $\Nodesub(\dectree)$. 

The second subdivision is given by the collection of polyhedral cones arising as those input vectors $x \in \RR^m$ that lead to the same leaf node of $\dectree$. 
Note that those sets are defined by strict and non-strict inequalities. 
We call this the \emph{leaf subdivision}. 

Finally, for the coarsest subdivision, consider the function from $\RR^m$ to $\{0,1\}^n$ represented by~$\dectree$. 
The preimage of each label $\ell \in \{0,1\}^n$ arises as the union of these cones leading to a leaf with label $\ell$. 
Note that such a union is not necessarily convex anymore. 
As we get one such region for each element arising as image in $\{0,1\}^n$, we refer to it as the \emph{range subdivision}. 

For us, the aim is, given a weight vector $w \in \RR^E$, to compute the winning region of player Max (or equivalently its complement, the winning region of player Min), that is a subset of $\Vmax\cup \Vmin$ encoded by a vector in $\{0,1\}^{\Vmax\cup \Vmin}$. 

The range subdivision is hard to understand as it basically corresponds to solving an MPG. 
On the other hand, the leaf subdivision heavily depends on the particular choice of LDT for an MPG. 
Our main contribution is about the node subdivision which captures general insights in the linear decisions which are necessary to derive the winning region. 
To solve mean payoff games using this model, we have to assign a linear decision tree to each possible arena $G=(\Vmax\cup \Vmin, E)$. 

\begin{definition}
Let $\mathcal{G}$ be the set of all possible arenas. 
An \emph{LDT algorithm for MPGs} is a set of linear decision trees $\setOf{\dectree_G}{G\in \mathcal{G}}$ such that $\dectree_G$ solves each MPG with arena $G$ given the weights of the edges as input. 
\end{definition}

It turns out that several pseudopolynomial algorithms for determining the zero-mean partition in a mean payoff game fall into the framework of the linear decision tree model. 
We make this explicit for some concrete examples. 
For the specific version of strategy improvement algorithms we have chosen \cite{DBLP:journals/dam/BjorklundV07,DBLP:journals/ijfcs/BrimC12}, 
for value iteration algorithm we refer to \cite{DBLP:journals/fmsd/BrimCDGR11} 
and furthermore we consider the `GKK algorithm' from \cite{gurvich1990cyclic}. 

Note that strategy iteration and value iteration usually have quite some flexibility in choosing which nodes or edges to update/improve first. 
Especially for strategy iteration, this gives rise to many improvement rules while this kind of flexibility is ruled out in the linear decision tree model. 
However, many natural improvement rules can still be represented by an LDT. 
For example, fixing Bland's least index rule allows to formulate the run of the algorithm purely in terms of linear decisions. 

\begin{proposition}
The strategy improvement algorithms from \cite{DBLP:journals/dam/BjorklundV07,DBLP:journals/ijfcs/BrimC12}, the value iteration algorithm from \cite{DBLP:journals/fmsd/BrimCDGR11} and the GKK algorithm from \cite{gurvich1990cyclic} can be implemented as LDT algorithms.
\end{proposition}
\begin{proof}[Proof sketch]
    We treat each type of algorithm separately:
    \begin{itemize}
        \item The main idea of the strategy improvement algorithms is to fix some strategy $\sigma$ for player Max, and compute some \emph{valuation} $d_{\sigma}:V\to \mathbb{R}$. The valuation is always equal to the weight of some shortest/longest path between two nodes. 
        The valuation is found with some number of shortest path computations (Dijkstra/Bellman-Ford). 
        Then, the main idea of strategy improvement is to change the strategy: we switch to using the edge $(u,v)$ if $d_{\sigma}(v)+w(u,v)>d_{\sigma}(u)$, we call this an \emph{improving move}. 
        Assuming a fixed improvement rule for choosing among possibly several improving moves, this yields a new strategy $\sigma'$.
        We repeat this process of finding valuations and making improving moves until there are no improving moves left. This gives an optimal strategy.

        These shortest path algorithms can also run in the LDT model, since we can compare path lengths, and we can store the valuations as sets of edges. If we do this, we can also evaluate whether $d_{\sigma}(v)+w(u,v)>d_{\sigma}(u)$ in the LDT model. So both strategy improvement algorithms can be formulated as LDT algorithm.
        \item The value iteration algorithm of \cite{DBLP:journals/fmsd/BrimCDGR11} maintains a value function $f:V\to \mathbb{Z}$. Initially, $f$ is 0 everywhere, but it gets updated using \emph{lifting operations}. To perform a lifting operation for a node $u$, we need to evaluate whether $f(u)<f(v)-w(u,v)$ and find the sign of $f(v)-w(u,v)$ for all edges of the form $(u,v)$. The values of the function $f$ always represent the weight of some finite walk in the graph. If the value of $f$ reaches above some threshold (which is bounded by the sum of absolute values of edge weights), it is set to $\infty$. If we reach a fixed point of the lifting operation, then we can find the zero-mean partition from the set of nodes with value $\infty$.

        Since the $f(v)$ represent the weight of some walk, we may store them as multisets of edges. This allows us to evaluate whether $f(u)<f(v)-w(u,v)$ and $f(v)-w(u,v)>0$ within the LDT model. Also, since we know the signs of all edge weights, we can also use sums of absolute values of edges in the LDT model. Hence we can also check whether $f(v)$ is larger than $\sum_{e\in E}|w(e)|$. So this value iteration algorithm is also an LDT algorithm.
        
        We note that implementing this algorithm naively can lead to the decision tree having infinite depth. However, one can use a modification that allows for using any real numbers, as described in \cite[p. 37]{dorfman2024improved}  
        which guarantees that the algorithm terminates in $O(m2^{n/2})$ iterations\footnote{The original source states $O(m2^n)$. We clarified with the author that this was a typo.}.
        \item The GKK algorithm is slightly more complicated. It maintains a potential function \break $\epsilon:V\to \mathbb{Z}$, and performs a \emph{potential transformation} on the edge costs: this replaces $w(u,v)$ by $w'(u,v)=w(u,v)+\epsilon(u)-\epsilon(v)$. The algorithm then computes the set of nodes $L$, from which player Min can guarantee seeing a negative weight edge before seeing a positive weight edge. Afterwards, it computes $\delta$, which is the smallest number such that increasing $\epsilon$ by $\delta$ on the current set $L$ will change $L$ in the next iteration. Then it changes $\epsilon$ (by adding $\delta$ to $\epsilon$ on $L$), and it repeats this until $\delta=\infty$. In that case, $L$ gives us the zero mean partition.

        We argue that $\delta$, $\epsilon(e)$ and $w'(e)$ can always be written as linear combinations of edge weights. 
        First, we start with $\epsilon=0$ everywhere, and the original edge weights. 
        Computing the set $L$ can be done by testing signs of edges. Finding $\delta$ is always done by taking the minimum of weights $w'(e)$ of a certain set of edges (determining that set is done by checking signs of $w'(e)$). These adapted weights $w'$ are always linear combinations of $w$, and therefore $\delta$ can also be written as a linear combination of weights. Finally, adding $\delta$ to $\epsilon$ on some nodes shifts some edge weights $w'(e)$ by $\delta$, which still keeps it a linear combination of $w$. We conclude that indeed, these variables can always be written as a linear combination of $w$, 
        hence one can implement this algorithm as an LDT algorithm. Note that it is also guaranteed to end in $O(2^n)$ iterations \cite{Ohlmann2022GKK}.     \qedhere
    \end{itemize}
\end{proof}

In the previous sections we saw that it is not always possible to reduce the weights to less than exponential while preserving the cycle pattern. 
We generalize this idea to the underlying structure of arbitrary LDT algorithms. 
The main question we try to answer in this section is the following: is it possible to improve the pseudopolynomial upper bound of these algorithms,
by using some preprocessing on the weights, that preserves the structure used by the algorithm? 
The answer will be no. 
This is a first step towards proving a lower bound on the running time of algorithms representable by an LDT using the minimal representation size of a cycle pattern. 

We start by formalizing the preprocessing of the weights. 

\begin{definition}
A \emph{core preserving reduction} for an LDT algorithm $\setOf{\dectree_G}{G\in \mathcal{G}}$ for the mean payoff games on the set of arenas $\mathcal{G}$ is a set of functions $\setOf{f_G}{G\in\mathcal{G}}$ with $f_{G} \colon \mathbb{Z}^E\to\mathbb{Z}^E$ with the following property: 
for each graph $G\in \mathcal{G}$, each vector $z\in \mathbb{Z}^E$ and each cone $F\in \Nodesub(\dectree_G)$, we have $z\in F\Leftrightarrow f_G(z)\in F$ and $\norm{f_G(z)}_{\infty} \leq \norm{z}_{\infty}$.
\end{definition}
In \cref{sec:cycle-patterns-arrangements}, we show that the fan $\Nodesub(\dectree_G)$ must be a refinement of the hyperplane arrangement coming from the cycle pattern. 
So the cycle pattern is in some sense the coarsest structure possible for an LDT algorithm. 
We then show that for any such hyperplane arrangement, there is an arena where one of the cells only has exponential integer weights. 
This implies that using only core preserving reductions, we cannot make a pseudopolynomial upper bound on the running time of an LDT algorithm that is less than exponential in the input size.

\subsection{Cycle patterns are coarsest}
\label{sec:cycle-patterns-arrangements}

Using our previous insights on cycle patterns and the framework of LDTs, we derive two insights about the size of weights in the different subdivisions. 
These are in some sense complementary to each other. 
First, we consider the range subdivision.

We fix an arena $G=(\Vmax \cup \Vmin, E)$ and study the set of weighted arenas $(G,w)$ for varying~$w$. 
For any subset $U\subseteq V(G)$, let
\[
\region{U} = \{w \in \RR^E \colon \val_i(w) \ge 0 \text{ for $i \in U$}, \val_i(w) < 0 \text{ for $i \notin U$} \}
\]
be the weights giving rise to the winning region $U$. 
Ranging over the different sets $\region{U}$, we get a subdivision of the space $\RR^m \simeq \RR^E$. 
With this terminology, we may interpret the zero-mean partition problem geometrically: given some $w\in \ZZ^E$, find $U$ such that $w\in \region{U}$. 
This subdivision of space agrees with the range subdivision for a linear decision tree $\dectree_G$. 

While the cones $\region{U}$ forming this subdivision are in general not polyhedral, not even convex, they turn out to be \emph{star-convex} sets. 
Recall that a set $S$ is star-convex if there exists a so-called \emph{center} $x_0 \in S$ such that for any $y\in S$ the line segment $x_0y$ is contained in $S$.

\begin{proposition}
    \label{lem:starconvex}
For any arena $G$ and vertex set $U$, if the set $\region{U}$ is nonempty, then it is star-convex.
Moreover, then there is a center of $\region{U}$ with entries in $\{-1,1\}$. 
\end{proposition}
\begin{proof}
    Let $z \in \RR^E$ be the weight vector defined by:
    \[
    z_{(v,v')}=
    \begin{cases}
        1 & v\in U,\\
        -1 & v\notin U.
    \end{cases}
    \] 
We want to show that $z$ is a center for $\region{U}$. 
For an arbitrary $w\in \region{U}$ and $\lambda \in [0,1]$, we need to show that $y=\lambda w+(1-\lambda) z$ is contained in $\region{U}$.
Since $w\in \region{U}$, there exists a Max strategy $\sigma$ that achieves an outcome $\geq 0$ starting in $U$ in the game on $(G,w)$, and achieves outcome $<0$ when starting in $V\backslash U$. 
In particular this means that there is no edge in $E_{\sigma}$ allowing Min to go out of $U$. 
Then, since all edge weights
within $U$ for $z$ are $1$, Max can achieve an outcome of $1$ on $U$ by playing strategy $\sigma$ on the weighted arena $(G,z)$.  
Now consider an arbitrary cycle $C$ whose vertices are contained in $U$ and whose edges are in $E_{\sigma}$. 
Since $\sigma$ ensures that the outcome is nonnegative on both $(G,w)$ and $(G,z)$, the total weight of $C$ must be nonnegative for both $w$ and $z$. 
Therefore, the weight of $C$ is nonnegative for any convex combination of $w$ and~$z$. 
Since this holds for any such cycle, the strategy $\sigma$ yields a nonnegative outcome starting from each node of $U$, in the game played on $(G,y)$.
Analogously, we can take a Min strategy $\tau$ that gives outcome $<0$ on $V\backslash U$ in $(G,w)$, and derive that it also gives negative outcome on $V\backslash U$ in $(G,y)$. 
We conclude that indeed $y\in \region{U}$.
Hence, the cone $\region{U}$ is star-convex with center $z$.
\end{proof}

While the lemma above shows that there is a `direct' way to go from arbitrary weights to (very!) small weights preserving the zero-mean partition, it is not clear how to use this algorithmically, since we used the zero-mean partition to find this weight reduction. 
We now relate the structure of the sets $\region{U}$ with cycle patterns.

\begin{proposition}\label{lem:ZUboundary}
    Assume $G$ is strongly connected. 
    For any cycle $C$, there is a set of vertices~$U$ and a weight vector $w$ 
    such that the intersection of the ball $B(w,1/2) = \setOf{t \in \RR^E}{\norm{w-t}_{\infty} < 1/2}$  
    with the boundary of $\region{U}$ is contained in the hyperplane 
    $
    H_C := \setOf{x \in \RR^E}{\sum_{(i,j)\in C}x_{ij}=0} \enspace .
    $
\end{proposition}

\begin{figure}[H]
    \centering
    \includegraphics[width=0.8\linewidth]{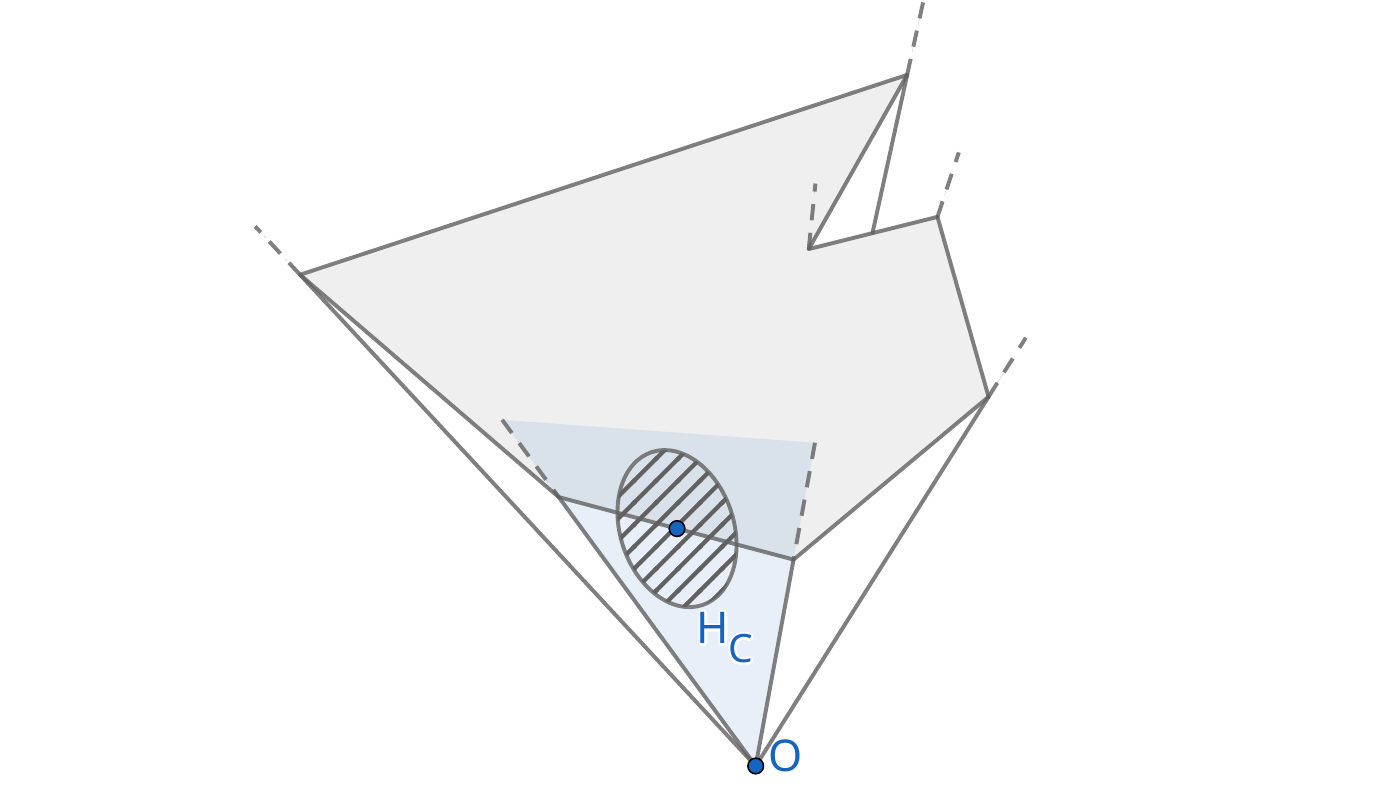}
    \caption{A cross section of a cone $\region{U}$. The intersection of the boundary of $\region{U}$ with $B(w, \frac{1}{2})$ (highlighted with stripes) lies on the hyperplane $H_C$ (blue).}
    \label{fig:cycle_hyperplane_lemma}
\end{figure}

\begin{proof}
Let $V(C)$ be the vertex set of $C$, 
and let $U=V$ be the whole vertex set of $G$. 
Furthermore, let $F$ be a shortest path forest containing the edges from shortest paths (in terms of number of edges) from every node in $V\backslash V(C)$ to $V(C)$. 
Note that $F$ exists because $G$ is strongly connected. 
We define the weight vector $w$ by 
    \[
    w_{(v,v')}=\begin{cases}
        0 & (v,v')\in C\cup F, \\
        |V| & v\in \Vmin, (v,v')\notin C\cup F, \\
        -|V| & v\in \Vmax, (v,v')\notin C\cup F \enspace .
    \end{cases} 
    \]
Now consider the weight vector $w+y$ for some $\norm{y}_{\infty}<\frac{1}{2}$. 
We show that $w+y\in \region{U}$ if and only if $\sum_{(i,j)\in C}(w+y)_{ij}\geq 0$.

Suppose that $\sum_{(i,j)\in C}(w+y)_{ij} \ge 0$. Since the arena is strongly connected, every vertex has an outgoing edge in $C \cup F$. If player Max plays according to some policy that only uses these edges, then they guarantee that the outcome of the game played on $(G,w+y)$ is nonnegative from any starting position. Indeed, let $\sigma$ denote any such policy and let $C'$ be any cycle composed of the edges of $E_{\sigma}$. If $C' = C$, then $C'$ has nonnegative weight by assumption. Otherwise, note that the cycle $C'$ has two types of edges: edges from $C \cup F$ which have total weight at least $-\frac{1}{2}|V|$, and edges of the form $(v,v') \notin C \cup F$ where $v \in \Vmin$. Every edge of the second type has weight at least $|V| - \frac{1}{2} > \frac{1}{2}|V|$ and $C'$ has at least one such edge because $C' \neq C$. Hence, $C'$ has positive weight. In particular, we have $w + y \in \mathcal{Z}(U)$. In the same way, if $\sum_{(i,j)\in C}(w+y)_{ij} < 0$, then player Min has a policy that guarantees that the outcome of the game is strictly negative from any starting position and $w + y \notin \mathcal{Z}(U)$. In conclusion, if a point inside $B(w,1/2)$ belongs to the boundary of $\mathcal{Z}(U)$, then it must lie on the hyperplane $\sum_{(i,j)\in C}(w+y)_{ij} = 0$.
\end{proof}

\begin{remark}
    If $G$ is not strongly connected then there can exist cycles that are not `critical' (in the sense that their hyperplane is not part of any boundary). For example, in Figure \ref{fig:nonstrong} the weight of cycle $ab$ has never any influence on which node is winning: that is because node $c$ always has at least as good of a value for one of the players as cycle $ab$, hence one player would always choose to leave the cycle $ab$.
    \begin{figure}[h]
        \centering
        \includegraphics[width=0.5\linewidth]{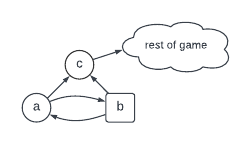}
        \caption{An arena that is not strongly connected. Circles are in $\Vmax$, squares in $\Vmin$.}
        \label{fig:nonstrong}
    \end{figure}
\end{remark}

\Cref{lem:ZUboundary} shows that the structure of the cycle pattern appears crucially in the node fan $\Nodesub(\mathcal{T}_G)$ associated with a linear decision tree for the game on $G$. 
This has some implications for core preserving reductions and pseudopolynomial algorithms. 
To make this precise, we now introduce a geometric hardness measure.

Given an arena $G$ and an LDT algorithm $\mathcal{T}_G$, we define the 
\emph{representability measure}
to be
\newcommand{\hardmeasure}{\texttt{RM}}
\begin{equation*}
\hardmeasure(\mathcal{T}_G) = \max_{z\in\mathbb{Z}^m}\min \setOf{ \norm{f_G(z)}_{\infty}}{f_G \text{ core preserving reduction}} \enspace .
\end{equation*}
By definition, this quantity provides a limitation for the effectiveness of core preserving reductions, but it can also be interpreted from other perspectives.

Firstly, we describe the geometric point of view on this quantity.
Suppose we take the smallest integer vector (w.r.t. the $\infty$-norm) in every cone of the node subdivision $\Nodesub(\mathcal{T}_G)$. 
Such a vector exists in every cone since they are halfopen polyhedra defined with rational inequalities. Then the representability measure is the norm of the largest of these integer vectors.

 Secondly, the representability measure can be related to the runtime of pseudopolynomial LDT algorithms. Suppose we have some pseudopolynomial LDT algorithm $\setOf{\mathcal{T}_{G}}{G\in\mathcal{G}}$ whose runtime is bounded by $p(W,|E|)$, with $p$ a polynomial and $W=\max_{e\in E}|w_e|$. For each weight vector $z\in \mathbb{Z}^E$ and each core preserving reduction $f_{G}$, the LDT algorithm makes the exact same choices for the game $(G,z)$ as for $(G,f_{G}(z))$. In particular we may use the best core preserving reduction $f_{G}$ and always assume that our algorithm solves the game $(G,f_{G}(z))$ instead of $(G,z)$. This then gives us a runtime bound of $p(\hardmeasure (\mathcal{T}_{G}),|E|)$, and in particular if $\hardmeasure (\mathcal{T}_{G})$ is bounded by a polynomial of the arena size, this means that the LDT algorithm solves any game on the arena $G$ in polynomial time.

However, we show that LDT algorithms for mean payoff games all have a large representability measure for some arenas.
Let again $\mathcal{G}$ be the set of all possible arenas. 

\begin{theorem} \label{thm:core-preserving-lower-bound}
There exists a sequence of arenas $(G_i)_{i\in \mathbb{N}}$ with strictly increasing number of edges~$(m_i)_{i\in \mathbb{N}}$ and a constant $c>0$, such that for any $i\in \mathbb{N}$ and any LDT algorithm $\setOf{\mathcal{T}_G}{G\in \mathcal{G}}$ 
we have 
\[
\hardmeasure(\mathcal{T}_{G_i})> 2^{c m_i} \enspace .
\]
\end{theorem}
\begin{proof}
From \cref{lem:ZUboundary}, we know that for any arena $G$ and any LDT algorithm, the node subdivision must contain the hyperplane $H_C=\setOf{x}{\chi(C)^Tx=0}$ for any cycle $C$ in the set of cycles $\Cs$ of $G$. 
Thus, the polyhedral fan $\Nodesub(\mathcal{T}_G)$ must be a refinement of the polyhedral fan arising from the hyperplanes $\setOf{H_C}{C\in \Cs}$. 
Otherwise, the LDT algorithm would give the wrong answer on some inputs, in particular those that are close to a multiple of the vector $w$ from \cref{lem:ZUboundary}.
From Theorem \ref{thm:expweights}, we know that there is a sequence of graphs $(G_{i})_{i\in\mathbb{N}}$ with realizable cycle patterns that require integer weights of size $2^{\Omega(m_i)}$. 
This implies, for each $G_i$, that the polyhedral fan coming from $\setOf{H_C}{C\in \Cs}$ has cells 
in which every integer vector has an entry that is exponential in the number of edges. 
Therefore, the same must hold for the finer polyhedral fans $\Nodesub(\mathcal{T}_{G_i})$. 
Let $z$ be an integer vector in such a cell only containing vectors with norm more than $2^{cm_i}$. 
Then, there cannot be a core preserving reduction that maps entries from this cell to smaller entries than $2^{cm_i}$. 
Hence $\min \setOf{ \norm{f_G(z)}_{\infty}}{f_G \text{ core preserving reduction}}>2^{cm_i}$, which implies the theorem. 
\end{proof}

We also extract the geometric aspect of the former proof. 

\begin{corollary}
    Any fan refining the hyperplane arrangement arising from the cycles of a complete graph has a cone containing only exponentially big integer points.  

    More precisely, let $(G_i)_{i\in \mathbb{N}}$ be the sequence of complete directed graphs indexed by its number of vertices and $(\Cs_i)_{i\in \mathbb{N}}$ be the corresponding sequence of sets of cycles. 
    Furthermore, let $(\mathcal{F}_i)_{i\in \mathbb{N}}$ be a sequence of polyhedral fans for which $\mathcal{F}_i$
    refines the polyhedral fan arising from the hyperplanes $\setOf{H_C}{C\in \Cs_i}$. 
    Then there is a sequence of cones, one for each element in $(\mathcal{F}_i)_{i\in \mathbb{N}}$, in which every integer vector has an entry that is exponential in the number of edges. 
\end{corollary}

\begin{remark}
    The paper \cite{Loho:2016} deduced a combinatorial and pseudopolynomial algorithm for mean payoff games based on the structure of the subdivions of a product of two simplices. 
    It turns out that the secondary fan of a product of two simplices also arises as refinement of the hyperplane arrangement associated to the cycle pattern. 
    This follows as the secondary fan can be described by inequalities on alternating cycles in a bipartite graph. 
    Therefore, there are regular triangulations of a product of two simplices $\Delta_{d-1} \times \Delta_{n-1}$ such that each integral weight vector inducing such a triangulation has an entry exponential in $d$ and $n$. 
    This extends the bounds given in \cite{BabsonBillera:1998}. 
    We leave it as an interesting research direction to deduce such results for other secondary fans of polytopes. 
\end{remark}

\section{Discussion and further thoughts}
\subsection{Extended cycle patterns}

We finish with a possible extension of our definition of cycle patterns. 
Given a digraph $G$, let $\Cs'$ be the set of ordered pairs of paths $(P_1,P_2)$, with the properties that $P_1$ and $P_2$ have the same start and end vertex, and that $P_1$ and $P_2$ share no internal vertices. 
Let $\Cs^E=\Cs\cup \Cs'$ be the union of these pairs of paths and of directed cycles. 
We define an \emph{extended cycle pattern} as a function $\psi \colon \Cs^E\to \{-,0,+\}$ on this extended domain, also associating a sign with pairs of paths. 
These signs indicate which of the paths has a larger weight.

More specifically, we say that $w$ is a \emph{realization} of $\psi$ if for every $C\in \Cs$ we have $\sgn(w(C))=\psi(C)$, and for every $(P_1,P_2)\in \Cs'$, we have $\sgn(w(P_1)-w(P_2))=\psi(P_1,P_2)$. 
The notions of realizability and induced extended cycle pattern are defined similarly. 
Note that, for every realizable extended cycle pattern and every $(P_1,P_2)\in \mathcal{C'}$, we have that $\psi(P_1,P_2)$ is the reverse sign of $\psi(P_2,P_1)$. 
It turns out that the extended cycle pattern has some nice computational advantage compared to the regular cycle pattern.

\begin{lemma}
    Given a realizable extended cycle pattern $\psi$ as a Boolean circuit, we can perform the following tasks in polynomial time:
    \begin{itemize}
        \item Decide if a graph contains negative cycles, and find one if it exists.
        \item Find a shortest directed walk between any pair of points, if one exists. 
    \end{itemize}
\end{lemma}
\begin{proof}[Sketch of proof]
    Since we can compare any two paths with disjoint interiors, we can use an adapted version of the Bellman-Ford algorithm to solve these two problems: in this adaptation, instead of storing a number for each node, we store a path for each node representing the shortest path found so far.
\end{proof}
\begin{corollary}
    Given a realizable extended cycle pattern $\psi$ of the graph, the zero-mean partition in one-player MPGs can be computed in polynomial time. 
\end{corollary}
This is different from regular cycle patterns as we saw in \cref{thm:1pZMP}. However, this is much closer to the properties of many algorithms for games on graphs. For example, many strategy improvement algorithms, the GKK algorithm and the value iteration algorithm from \cite[p. 37]{dorfman2024improved} also solve one-player games in polynomial time. For some strategy improvement algorithms, one-player computation is used as a (polynomial time) subroutine, e.g.  \cite{DBLP:journals/dam/BjorklundV07}, and the other two algorithms also solve one-player games in polynomial time, as direct consequences of \cite[Thm. 5]{Ohlmann2022GKK} and \break \cite[Lem. 2.2.14]{dorfman2024improved}, respectively.

This is an indication that the extension allows for a more precise analysis of algorithms. 
Even stronger, the strategy improvement algorithm from \cite{DBLP:journals/dam/BjorklundV07} only uses shortest path computations and comparisons of path lengths. 
Adapting it to the setting of extended cycle patterns is left for future work. 

\subsection{Discussion}
We have defined the cycle pattern, characterized realizability, and have shown that several problems are hard to compute when only given the cycle pattern as a Boolean circuit. In this case much of the complexity comes from the encoding of the cycle pattern. However, it shows that, if a (randomized) polynomial time algorithm for the zero-mean partition problem existed, it must somehow utilize information from a specific realization (assuming $NP\neq RP$).

The hardness of parity-realizability has an interesting consequence.
While some attempts have been made to translate ideas from parity games to mean payoff games, like \emph{quasi-dominions} \cite{benerecetti2024quasi,Fearnley2010Nonoblivious} and \emph{universal graphs} \cite{fijalkow:hal-03800510}, they did not turn out to give a significant algorithmic advantage. 
Our result may be considered a complexity theoretic quantification of this phenomenon, showing how different the structure of parity games and mean payoff games is, given that cycle patterns capture a fundamental essence of the algorithmically relevant structure. 
Stated more boldly, we get the following restriction for applying parity game techniques to MPGs: Suppose we are able to adapt a quasipolynomial parity game algorithm to work on MPGs in the same way. 
Then, for every instance of a MPG whose cycle pattern is parity-realizable, our algorithm cannot detect this fact unless $NP=QP$.
Here $QP$ is the class of problems that can be solved in quasipolynomial time (there exists some $c\in \mathbb{N}$, such that it can be decided on a Turing machine in time $O(2^{\left(\log(n)\right)^c})$, where $n$ is the input size).

We have shown that there are graphs for which the representability measure is large for any algorithm. 
This does not imply large running time, although it does provide an explanation why these algorithms have exponential worst-case behaviour, especially for pseudopolynomial algorithms. This explanation is more powerful than many well-known graph complexity measures, because it also holds for graphs where many commonly used ones are small (see \cref{thm:expweights}). We leave an explicit connection between sizes of realizations and complexity of (LDT) algorithms as further work.

It would also be possible to define the representability measure for the leaf subdivision. 
It is still open how large it can be in that case. 
Finally, as already mentioned, complexity of the bounded realization problem is also still open, as is the question if the lower bound from \cref{thm:expweights} can be improved to $2^{\Theta(m\log (n))}$.

We hope this is just the beginning of a geometric complexity theory for mean payoff games.

\bibliography{patterns}

\end{document}